\documentclass[final]{IEEEtran}

\usepackage{ifpdf}
\ifpdf
	\usepackage[pdftex]{graphicx}
  \usepackage[update]{epstopdf}
\else
	\usepackage{graphicx}
\fi
\usepackage{amsmath,amssymb,array,stfloats,cuted,midfloat,psfrag,mathbbol,mathrsfs,setspace}

\newtheorem{theorem}{\textbf{Theorem}}

\newtheorem{definition}{\textbf{Definition}}%[section]

\newtheorem{lemma}{\textbf{Lemma}}

\DeclareMathAlphabet{\mathitbf}{OML}{cmm}{b}{it}
\renewcommand{\vec}{\mathitbf}
\renewcommand{\vec}[1]{\mbox{\boldmath$#1$}}

\begin{document}

\date{}

\title{Latent Capacity Region: A Case Study on Symmetric Broadcast With Common Messages}

\author{Chao Tian,~\IEEEmembership{Member,~IEEE}
\thanks{Chao Tian is with AT\&T Labs-Research, Florham Park, NJ 07932, USA (email: tian@research.att.com).}
}

\maketitle

\begin{abstract}
We consider the problem of broadcast with common messages, and focus on the case that the common message rate $R_{\mathcal{A}}$, i.e., the rate of the message intended for all the receivers in the set $\mathcal{A}$, is the same for all the set $\mathcal{A}$ of the same cardinality. Instead of attempting to characterize the capacity region of general broadcast channels, we only consider the structure of the capacity region that any broadcast channel should bear. The concept of latent capacity region is useful in capturing these underlying constraints, and we provide a complete characterization of the latent capacity region for the symmetric broadcast problem. The converse proof of this tight characterization relies on a deterministic broadcast channel model. The achievability proof generalizes the familiar rate transfer argument to include more involved erasure correction coding among messages, thus revealing an inherent connection between broadcast with common message and erasure correction codes.
\end{abstract}
\begin{keywords}
Broadcast channel, common message, individual message.
\end{keywords}

\section{Introduction}
\label{sec:intro}

One central theme in multi-user information theory (IT) is the pursuit of single-letter\footnote{The emphasis on single letter is largely because such kind of characterization is usually computable.} characterizations of the capacity regions for channel coding problems, or the achievable rate regions (possibly under certain distortion constraints) for source coding problems. However, some useful properties of these regions can be identified, e.g., convexity, even when a single-letter characterization is not available. An immediate question to ask is whether there exist other properties of the capacity region that do not rely on a single letter characterization. 

The following question is of interest in this regard: in a particular multi-user IT problem, can the achievability of a rate vector $(R^*_1,R^*_2,\ldots,R^*_N)$ imply the achievability of any rate vector in some region\footnote{Apparently the region defined by $R_i\leq R^*_i$ is implied in a channel coding problem, but this trivial case is not interesting. Note here we do not take the subscript of rate $R^*_i$ to have any specific meaning associated with the user indices, but merely as an integer label to enumerate the rates in question.} $\mathcal{R}(R^*_1,R^*_2,\ldots,R^*_N)$,  regardless of the exact probabilistic channel model? We show that indeed this is true for the symmetric broadcast problem, and this region can be rather non-trivial. We denote the largest of such regions $\mathcal{R}(R^*_1,R^*_2,\ldots,R^*_N)$ as $\mathcal{C}(R^*_1,R^*_2,\ldots,R^*_N)$ in a channel coding problem, and call it the \textit{latent capacity region} implied by $(R^*_1,R^*_2,\ldots,R^*_N)$; the \textit{latent achievable rate region} can be similarly defined, possibly under certain distortion constraints, for a source coding problem though it is not our main focus. 

For broadcast and multiple access channels, a precise problem formulation was given in a recent work by Grokop and Tse \cite{GrokopTse:08}, called {\em multicast region}, which provides a framework to answer the above question. Complete solutions were found in \cite{GrokopTse:08} for broadcast and multiple access channels with \textbf{two} and \textbf{three} users, but the problem remains open for more than three users. We believe this problem formulation reveals a more general concept not limited to only these two channels, and thus rename it as the latent capacity (or latent achievable rate) region problem to make explicit this generality. Our perspective is different from \cite{GrokopTse:08} in that we wish to highlight the importance of the latent capacity region concept in its ``maximum implication" meaning, and thus we shall define the region in an alternative (but equivalent) manner to emphasize this perspective; our interest in this problem is partially due to an observation made during an earlier work \cite{TianDiggavi:07}, as we shall discuss shortly.

One may wonder how a single achievable rate vector $(R^*_1,R^*_2,\ldots,R^*_N)$ can imply the achievability of a certain region. In some cases, it is perhaps best explained by the familiar rate transfer argument, that the rate to transmit common messages can be used to transmit individual messages instead, and vice versa. For example, for a two user broadcast channel, if a common message rate $R^*_{\{1,2\}}$, and individual message rates $R^*_{\{1\}}$ and $R^*_{\{2\}}$ are achievable, respectively, then it is not difficult to see that the region of $(R_{\{1,2\}},R_{\{1\}},R_{\{2\}})$ given below is achievable by transferring between common and individual rates (see also \cite{GrokopTse:08})
\begin{align*}
R_{\{1,2\}}+R_{\{1\}}&\leq R^*_{\{1,2\}}+R^*_{\{1\}}\\
R_{\{1,2\}}+R_{\{2\}}&\leq R^*_{\{1,2\}}+R^*_{\{2\}}\\
R_{\{1,2\}}+R_{\{1\}}+R_{\{2\}}&\leq R^*_{\{1,2\}}+R^*_{\{1\}}+R^*_{\{2\}}.
\end{align*}
However, for more than two users, such a naive rate transfer argument is not sufficient, and additional processing is needed, as observed in \cite{GrokopTse:08} for the three user case. In fact, this was exactly the perspective taken in \cite{GrokopTse:08}, where the goal is to exhaust all such rate transfer operations. The perspective taken in \cite{GrokopTse:08} and that taken here are complementary to each other, and one may suit certain problems better than the other. 
Because of this relation, it is not surprising that the achievability proof of our result also relies on a generalized version of rate transfer operations. We shall show that when more users are involved, such generalized rate transfer operation requires strategic application of erasure correction codes, which reveals an inherent connection between erasure correction codes and broadcast with common messages. More specifically, in this work, we shall largely stay in the framework of \cite{GrokopTse:08}, and provide a complete solution to the $K$-user broadcast channel latent capacity region problem under an additional symmetry constraint, whereas only cases with two and three users were solved in \cite{GrokopTse:08} without such a constraint. 

\begin{figure}[tb]
\begin{centering}
\includegraphics[width=8.5cm]{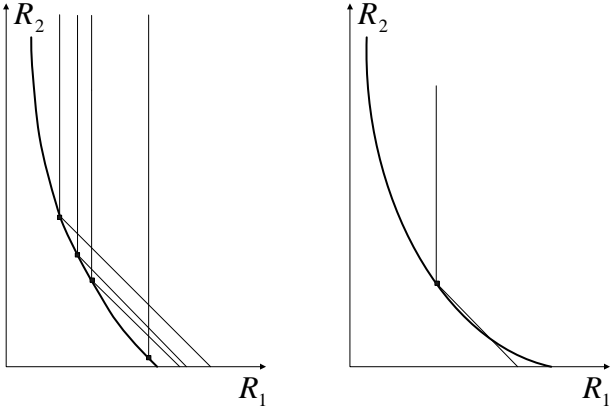}
\caption{The bold curve gives the rate region: while the left one is possible, the right one is  impossible for the successive refinement source coding problem. The thin lines give the latent capacity region associated with each small black dot.\label{fig:twoshapes}}
\end{centering}
\end{figure}

The characterization of latent capacity/rate region is important in multi-user IT for two reasons. First, it may facilitate finding a single-letter characterization or an approximate characterization. For example, a rate-distortion region characterization for the problem of  multi-stage successive refinement with degraded decoder side information was given in the form of bounds on sum-rates \cite{TianDiggavi:07} as
\begin{align}
&\sum_{i=1}^m R_i\geq \sum_{i=1}^m I(X;W_m|W_1,W_2,...,W_{m-1},Y_m),\nonumber\\
&\qquad\qquad\qquad\qquad\qquad\qquad 1\leq m \leq N. \notag
\end{align}
On the other hand, it seems impossible to establish directly the converse for a characterization in the form of bounds on each individual incremental rate \cite{SteinbergMerhav:04}, despite the fact that the two characterizations are equivalent \cite{TianDiggavi:07}. This is not a coincidence, and it is not difficult to show that the latent rate region for this problem has exactly the following form, assuming non-negativity of the rates
\begin{align}
&\sum_{i=1}^m R_i\geq \sum_{i=1}^m R^*_i, \quad 1\leq m \leq N.
\end{align}
Intuitively, when a rate vector $(R^*_1,R^*_2,\ldots,R^*_N)$ is achievable, its latent capacity/rate region gives the largest achievable region thus implied, i.e., maximally utilizes it, which may help simplify the representation of the region when taking the union over auxiliary random variables. Similarly, when an approximate characterization is needed, a good inner bound may be found by choosing one or several good (auxiliary coding) distributions in an information theoretic coding scheme which lead to one or several rate vectors, and then taking the convex hull of their latent rate regions. One simple example is in \cite{TianDiggavi:08}, where an approximate characterization for the side-information scalable source coding problem was given for general sources under the squared error distortion measure, and the inner bound approximation is exactly the latent capacity region implied by a single rate pair.

The second reason making this concept important is even if it does not lead to a single-letter characterization or an approximate characterization, it can still provide insights into the problem. One such example is that the capacity region can always be written as the (possibly uncountable) union of latent capacity regions, which places certain constraints on the geometry of the achievable region. For the above example of successive refinement source coding, we show in Fig. \ref{fig:twoshapes} a possible rate region on the left, and an impossible rate region on the right. The one on the right is impossible because the black dot is in the achievable region, thus the latent capacity region implied by it (given by the thin line) must be also in the region, which is not satisfied by the region depicted on the right. This important observation was also discussed in \cite{GrokopTse:08} (see Corollary 4.3), and we do not elaborate it further. Nevertheless, it is rather clear that the latent capacity region indeed provides fundamental and useful property of the rate region, in addition to the well-known convexity.

\section{Problem Definition and Preliminaries}

We first define the symmetric broadcast problem, and then introduce the notion of latent capacity region in this context.
 
In a general $K$-user broadcast channel, the conditional probability distribution is given as 
\begin{align}
p(y_1[1,2,\ldots],y_2[1,2,\ldots],...,y_K[1,2,\ldots]\big{|}x[1,2,\ldots])
\end{align}
where the index in the bracket $[1,2,...]$ is used to denote time; the random variables have alphabets $\mathscr{X},\mathscr{Y}_1,\ldots,\mathscr{Y}_2$, and the receivers are indexed as $1,2,\ldots,K$. 
The alphabets can be discrete or continuous, and the channel can be memoryless or otherwise; for our purpose, it is perhaps beneficial, though not necessary, to limit the attention to cases where the channel transition process is (block) stationary and ergodic.  
We use script letters to denote sets, and particularly, $\mathcal{A}$ and $\mathcal{B}$ are reserved for subsets of $\mathcal{I}_K=\{1,2,\ldots,K\}$, i.e.,
\begin{eqnarray}
\mathcal{A},\mathcal{B}\subseteq \{1,2,\ldots,K\}.
\end{eqnarray}
$|\mathcal{A}|$ is used to denote the cardinality of set $\mathcal{A}$. A length-$n$ vector $X[1,2,...,n]$ is sometimes written as $X^n$; for a $K$ dimensional vector $(R_1,R_2,...,R_K)$, we sometimes write it simply as $\vec{R}$. 

Let $\{W_{\mathcal{A}},\mathcal{A}\subseteq\mathcal{I}_K\}$ be $2^{K}$ mutually independent and uniformly distributed messages, where $W_{\mathcal{A}}$ is the message intended for all the receivers in the set $\mathcal{A}$; for notational convenience, we include $W_{\emptyset}$ but will assume it to be a constant.
For each $k=1,2,\ldots,K$, define the set of random variables
\begin{align}
\mathcal{W}_k=\{W_{\mathcal{A}}, \mathcal{A}:k\in \mathcal{A}\}.
\end{align}
Thus $\mathcal{W}_k$ is the collection of messages that the $k$-th receiver should decode. We also define the following set of random variables
\begin{align}
\overline{\mathcal{W}}_k=\{W_{\mathcal{A}},|\mathcal{A}|\geq k\}.
\end{align}
More specifically for $K=3$, we have
\begin{align}
\mathcal{W}_1&=\{W_1,W_{12},W_{13},W_{123}\}\notag\\
\mathcal{W}_2&=\{W_2,W_{12},W_{23},W_{123}\}\notag\\
\mathcal{W}_3&=\{W_3,W_{13},W_{23},W_{123}\}\notag\\
\overline{\mathcal{W}}_1&=\mathcal{W}_1\cup\mathcal{W}_2\cup\mathcal{W}_3\notag\\
\overline{\mathcal{W}}_2&=\{W_{12},W_{13},W_{23},W_{123}\}\notag\\
\overline{\mathcal{W}}_3&=\{W_{123}\},
\end{align}
where we have slightly abused the notation by writing, e.g., $W_{\{1\}}$ as $W_1$. The sets $\mathcal{X}^n_i$ and $\overline{\mathcal{X}}^n_i$ are defined similarly for length-$n$ random vectors.
In this work, we only consider the case that the rates of messages $W_{\mathcal{A}}$ are the same for all such messages where the set $\mathcal{A}$ has the same cardinality. More formally, the problem is defined as follows.
\begin{definition}
An $(n,R_1,R_2,\ldots,R_K)$ symmetric broadcast code consists of an encoder 
\begin{align}
f: \prod_{\mathcal{A}\subseteq \mathcal{I}_K}\mathcal{I}_{2^{nR_{|\mathcal{A}|}}}\rightarrow \mathscr{X}^n,
\end{align}
where $R_{\emptyset}\triangleq 0$ and $K$ decoders,
\begin{align}
g_{k}:\mathscr{Y}^n_k\rightarrow \prod_{\mathcal{A}:k\in \mathcal{A}}\mathcal{I}_{2^{nR_{|\mathcal{A}|}}},
\end{align}
resulting in the decoded messages at the $k$-th receiver $\{\hat{W}_{k,\mathcal{A}}:k\in \mathcal{A}\}$, and the decoding error probability of at least one message at one receiver
\begin{align}
P^{(n)}_e=\mbox{Pr}\left(\bigcup_{k=1}^K \bigcup_{\mathcal{A}:k\in \mathcal{A}}\{W_{\mathcal{A}}\neq\hat{W}_{k,\mathcal{A}}\}\right).
\end{align}
\end{definition}
\begin{definition}
A rate vector $\vec{R}$ is symmetrically achievable if there exists a sequence of 
 $(n,\vec{R})$ codes with $P^{(n)}_e\rightarrow 0$. The closure of the set of symmetrically achievable rate vectors is called the symmetric broadcast capacity region, denoted as $\mathcal{C}_{p(y_1,y_2,\ldots,y_K|x)}$, or simply as $\mathcal{C}_{p}$.
\end{definition}

Note that secrecy constraint is not considered in the definition. Next we define the latent capacity region for this problem. 
\begin{definition}
\label{definition:LCR}
For a given rate vector $\vec{R}^*$, the collection of rate vectors $\mathcal{R}(\vec{R}^*)$ is called the latent capacity region for symmetric broadcast implied by $\vec{R}^*$, denoted as $\mathcal{C}(\vec{R}^*)$, if the following two conditions are satisfied (\textsf{i}) For any broadcast channel,  $\vec{R}^*\in\mathcal{C}_{p}$ implies $\mathcal{R}(\vec{R}^*)\subseteq \mathcal{C}_{p}$;
(\textsf{ii}) There exists a set of channels $\{p_x\}$, such that $\vec{R}^*\in\mathcal{C}_{p_x}$ and $\mathcal{R}(\vec{R}^*)\supseteq \bigcap_{x}\mathcal{C}_{p_x}$.
\end{definition}

For the second condition, we essentially wish to find one particular channel such that $\mathcal{R}(\vec{R}^*)\supseteq \mathcal{C}_{p}$. However this does not quite serve the purpose since this channel might be difficult to realize, however it can always be approximated by a sequence of channels. The above definition is slightly different from the one in \cite{GrokopTse:08}, which is
\begin{align}
\mathcal{C}(\vec{R}^*)=\bigcap_{p:\mathitbf{R}^*\in \mathcal{C}_p}\mathcal{C}_p.
\end{align}
It can be easily verified that they are equivalent. The problem we wish to solve is the characterization of $\mathcal{C}(\vec{R}^*)$. It is clear that the region $\mathcal{C}(\vec{R}^*)$ is uniquely defined for any $\vec{R}^*$, and thus the problem is meaningful. 

Definition \ref{definition:LCR} makes clear the ``maximal implication" meaning of the latent capacity region. In multi-user IT, usually a coding scheme is given by fixing some auxiliary random variables, and then showing a single rate vector is achievable with certain random codes; the task of maximizing the implication region of this single point is sometimes mingled with the conditions under which this single point is achievable. The concept of latent capacity region can be used to delineate them. 

The following lemma is needed in the converse proof. 
\begin{lemma}[$K$-way submodularity]
\label{lemma:Kwaysubmodularity}
Let $\{U_i,i=1,2,\ldots,N\}$ be a set of mutually independent random variables, and $\{V_i,i=1,2,\ldots,N\}$ be a set of random variables jointly distributed with it. Let $\mathcal{G}_i$, $i=1,2,\ldots,K$ be subsets of $\mathcal{I}_N$.
Then
\begin{align}
\sum_{k=1}^KH(V_i,i\in\mathcal{G}_k|U_i,i\in\mathcal{G}_k)\geq \sum_{k=1}^KH(V_i,i\in\widehat{\mathcal{G}}_k|U_i,i\in\widehat{\mathcal{G}}_k),
\end{align}
where
\begin{align}
\widehat{\mathcal{G}}_{k}\triangleq \bigcup_{\{j_1,j_2,\ldots,j_k\}\subseteq \mathcal{I}_K} (\mathcal{G}_{j_1}\cap\mathcal{G}_{j_2}\cap\cdots\cap\mathcal{G}_{j_k}).
\end{align}
\end{lemma}

This lemma is a direct consequence of the sub-modularity of the conditional entropy function, when the random variables being conditioned on are independent (a proof is given in Appendix \ref{append:Kwaysubmodularity}), and the $K$-way submodularity property of any submodular function given in  \cite{Harvey:06}.

\section{Main Result}

Our main result is a complete characterization of the latent capacity region for the symmetric broadcast problem. To present this region, a few more quantities need to be defined first. Let us define the following up-exchange rate for $i<j$ 
\begin{align}
\phi_{i,j}=\binom{K-i}{j-i}^{-1}\binom{j-1}{j-i},
\end{align}
and the down-exchange rate for $i>j$
\begin{align}
\phi_{i,j}=\binom{i}{i-j}^{-1}\binom{K-j}{i-j},
\end{align}
and define $\phi_{i,i}=1$. The up/down exchange rates $\phi_{i,j}$ essentially describe the ratio when converting certain type of messages into other types. For example when $K=3$, the common message $W_{123}$ can be used to convey individual information to the three users, and vice versa, but the conversion of such rates is not always ratio one. 
It will become clear in the achievable proof how such conversion can be done in a most efficient manner. 

Define $\mathcal{C}^*(\vec{R}^*)$ to be the set of rate vectors $\vec{R}$ satisfying the following conditions with some $K^2$ non-negative quantities $r_{i,j}$, $(i,j)\in \mathcal{I}_K\times\mathcal{I}_K$,
\begin{align}
R^*_i&\geq \sum_{j=1}^K r_{i,j},\quad i=1,2,\ldots,K,\label{eqn:condition1}\\
0\leq R_j&\leq \sum_{i=1}^K \phi_{i,j}r_{i,j},\quad j= 1,2,\ldots, K.\label{eqn:condition2}
\end{align}

Roughly speaking, the rate $r_{i,j}$ is that taken from level-$i$ rate $R^*_i$ but used to transmit level-$j$ messages. We have the following theorem.
\begin{theorem}
\label{theorem:maintheorem}
For any non-negative rate vectors $(R^*_1,R^*_2,\ldots,R^*_K)$, we have
\begin{align}
\mathcal{C}(R^*_1,R^*_2,\ldots,R^*_K)=\mathcal{C}^*(R^*_1,R^*_2,\ldots,R^*_K).
\end{align}
\end{theorem}

\textit{Example:} for $K=2$, it is straightforward to see: $\phi_{1,2}=1$, i.e., the same amount of individual message rate for each user can be used to transmit a common message; and $\phi_{2,1}=1/2$, i.e., to split a common message into two equal parts, each to transmit a separate individual message for one user.  

\textit{Example:} for $K=3$, it can be verified using Fourier-Motzkin elimination \cite{Ziegler} that $\mathcal{C}^*(R^*_1,R^*_1,\ldots,R^*_K)$ is given by the non-negative rates satisfying
\begin{align}
3R_1+6R_2+2R_3\leq 3R^*_1+6R^*_2+2R^*_3,\notag\\
2R_1+2R_2+1R_3\leq 2R^*_1+2R^*_2+1R^*_3,\notag\\
1R_1+2R_2+1R_3\leq 1R^*_1+2R^*_2+1R^*_3,\notag\\
3R_1+3R_2+1R_3\leq 3R^*_1+3R^*_2+1R^*_3.
\end{align}
A typical shape is given in Fig. \ref{fig:example122} with $(R_1^*,R^*_2,R^*_3)=(1,2,2)$.
The computation is tedious and thus omitted here. The same result can also be reduced from that given in \cite{GrokopTse:08} for the asymmetric case. It is clear that this region is non-trivial, and it is not at all clear a priori why these rate combinations should be considered.

In \cite{GrokopTse:08}, the region is characterized by investigating the distinct universal encoding/decoding operations, which leads to the concept of extremal rays. Because the latent capacity region in question is a polytope, it can be characterized by its faces, edges, or vertices. The extremal rays are essentially the edges of this polytope. However this proof approach in \cite{GrokopTse:08} appears rather difficult to generalize for more than three users since the number of edges quickly becomes very large, and thus we introduce the parametric characterization (\ref{eqn:condition1}) and (\ref{eqn:condition2}) to avoid this difficulty.

\begin{figure}[tb]
\begin{centering}
\includegraphics[width=8.5cm]{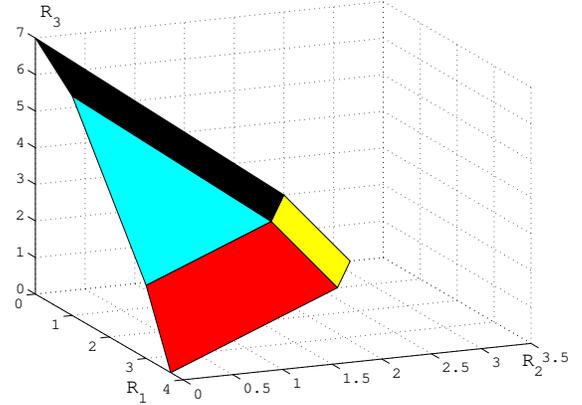}
\caption{The latent capacity region implied by rate vector $(1,2,2)$.\label{fig:example122}}
\end{centering}
\end{figure}

Notice that the exchange rate is pairwise, suggesting in this symmetric setting there is no need to convert rates jointly, e.g., use $W_{12}$ and $W_3$ to send the same message $W_{123}$. In the rest of the paper, we shall prove Theorem \ref{theorem:maintheorem}. The naive approach of finding the planes of the rate region and derive its upper and lower bounds is not appropriate for general $K$, particularly for the purpose of converse. Instead, we utilize the structure of the region $\mathcal{C}^*(\vec{R}^*)$ to give a proof.

\section{Proof of the Forward Part For Theorem \ref{theorem:maintheorem}}

The proof of the forward part of Theorem \ref{theorem:maintheorem}, i.e., the fact that $\mathcal{C}^*(\vec{R}^*)$ satisfies the first condition in Definition \ref{definition:LCR} is relatively straightforward. 

\begin{proof}[Proof of forward part for Theorem \ref{theorem:maintheorem}]
Since $(\vec{R}^*)$ is achievable on any channel, there exists a sequence of codes with such rates with $P^{(n)}_e\rightarrow 0$, and we will use these codes to construct a set of codes to approach any rate vectors in $\mathcal{C}^*(\vec{R}^*)$. This is done by essentially relabeling and adding erasure correction codes on the messages.  

Observe that the messages $\{W_{\mathcal{A}},|\mathcal{A}|=i\}$ can also be used to transmit common messages to the subsets with cardinality smaller or larger than $i$. Moreover, we can use part of the rate $R^*_i$, denoted as $r_{i,j}$, for this purpose, to transmit some messages $\{W'_{\mathcal{A}},|\mathcal{A}|=j\}$, thus increasing $R_j$. Such an operation will cause a conversion of rate $r_{i,j}$ for $i$-user subset messages into rate $\phi_{i,j}r_{i,j}$ for $j$-user subset messages, with an exchange rate $\phi_{i,j}$. The region $\mathcal{C}^*(\vec{R}^*)$ is precisely the result of allowing this kind of pairwise exchange on the rate vector $\vec{R}^*$. Thus we only need to show that the exchange rates $\phi_{i,j}$ given before Theorem \ref{theorem:maintheorem} is indeed valid, then the existing sequence of channel codes can be used directly. 

It is clear that we only need to consider the following problem: on a channel with $R_i=R$ and $R_k=0$ for $k\neq i$, how do we transmit messages $\{W_{\mathcal{A}},|\mathcal{A}|=j\}$, and how much rate $R_j$ can be supported? We will only need to distinguish two cases $i<j$ or $i>j$, since it is clear $\phi_{i,i}=1$. 

We first consider the case $i<j$. For a subset $\mathcal{B}$ of $\mathcal{I}_K$ where $|\mathcal{B}|=j$, there are a total of $\binom{j}{i}$ subset of $\mathcal{B}$ with cardinality $i$; denote the collection of such subsets as $2^{\mathcal{B},i}$. For a particular user $k\in \mathcal{B}$, it can decode (with high probability) the messages $\{W_\mathcal{A}:k\in \mathcal{A}\subset \mathcal{B}\}$, i.e., $\binom{j-1}{i-1}$ such messages. To transmit the common message $W_\mathcal{B}$, if we can guarantee that when receiving any $\binom{j-1}{i-1}$ messages out of the $\binom{j}{i}$ messages in the set $2^{\mathcal{B},i}$, the message is decodable, then it is clear that indeed any receivers in the set $\mathcal{B}$ can decode the message $W_\mathcal{B}$. This is an erasure correction problem and a $\left(\binom{j}{i},\binom{j-1}{i-1}\right)$ maximum distance separable (MDS) code can satisfy this requirement, which indeed exists when the codeword length is sufficiently large. Furthermore, since each subset $\mathcal{A}$ of cardinality $i$ is a subset of $\binom{K-i}{j-i}$ sets of cardinality $j$, only $\binom{K-i}{j-i}^{-1}$ of the rate $R_{\mathcal{A}}$ can be used for each MDS code. This yields 
\begin{align}
R_j=\binom{K-i}{j-i}^{-1}\binom{j-1}{i-1} R_i=\phi_{i,j}R_i.
\end{align}

Next consider the case $i>j$. Let $\mathcal{B}$ be a subset of $\mathcal{I}_K$ where $|\mathcal{B}|=i$. The common message $\mathcal{B}$ can be shared uniformly between its $\binom{i}{j}$ subsets of cardinality $j$, for transmitting their ``individual" message. Since each subset $\mathcal{A}$ of cardinality $j$ is a subset of distinct $\binom{K-j}{i-j}$ sets of cardinality $i$, it can take part in such sharing $\binom{K-j}{i-j}$ times. This yields
\begin{align}
R_j=\binom{i}{j}^{-1}\binom{K-j}{i-j}R_i=\phi_{i,j}R_i.
\end{align}
Taking into account the existence of good MDS code, and the fact that $\mathcal{C}_p$ is a closed set, the proof is complete. 
\end{proof}

In \cite{GrokopTse:08}, it was observed that in order to efficiently transfer rates, sometimes a modulo two addition is needed, similar to that seen in butterfly network of network coding \cite{Raymond:2000}. The MDS codes we use in the above proof can be understood as a generalization of the modulo two addition, which itself is essentially a $(3,2)$ MDS code. It is worth noting that other coding/processing may also be useful for converting rates, however, MDS codes are sufficient in solving the symmetric broadcast problem.

\section{Proof of the Converse Part For Theorem \ref{theorem:maintheorem}}
\label{sec:converse}
The converse proof of Theorem \ref{theorem:maintheorem} requires more work. For simplicity we shall assume $2^{R^*_{\mathcal{A}}}$'s are all integers; if this is not the case, a sequence of channels need to be considered, and we shall return to this technical point after the proof.

We only need to provide one particular channel that $\vec{R}^*\in\mathcal{C}_{p}$ and  $\mathcal{R}(\vec{R}^*)\supseteq \mathcal{C}_{p}$. The channel is the deterministic one considered in \cite{GrokopTse:08}, extended to the $K$-user case; see Fig. \ref{fig:deterministic} for the case $K=3$. More precisely, let the channel input be the collection of $\{X_{\mathcal{A}},\mathcal{A}\subseteq \mathcal{I}_K\}$. The alphabet of $X_{\mathcal{A}}$ where $|\mathcal{A}|=k$ is $\mathcal{I}_{2^{R^*_k}}$. The $k$-th channel output $Y_k$ is given by
\begin{align}
Y_k=\{X_{A}:k\in \mathcal{A}\}.
\end{align}
Denote this deterministic channel as $p^*$. 
In order to prove the converse part for Theorem \ref{theorem:maintheorem}, we need to establish $\mathcal{C}^*(\vec{R}^*)\supseteq \mathcal{C}_{p^*}$ for this channel. 

For any $\vec{A}=A_1,A_2,\ldots,A_K$ where $A_i\geq 0$, define the following quantity
\begin{align}
\label{eqn:maximizingBstar}
B_{\mathcal{C}^*}(\vec{A})=\max_{\mathitbf{R}\in \mathcal{C}^*(\mathitbf{R}^*)} \sum_{k=1}^K A_kR_k,
\end{align}
and similarly 
\begin{align}
\label{eqn:maximizingB}
B_{\mathcal{C}}(\vec{A})=\max_{\mathitbf{R}\in \mathcal{C}_{p^*}} \sum_{k=1}^K A_kR_k.
\end{align}
It is clear that both  $\mathcal{C}^*(\vec{R}^*)$ and $\mathcal{C}_{p^*}$ are convex regions, and thus if we can prove the following theorem, then the converse of Theorem \ref{theorem:maintheorem} directly follows. 
\begin{theorem}
\label{theorem:boundingplane}
For any $\vec{A}$ where $A_i\geq 0$,
\begin{align}
B_{\mathcal{C}^*}(\vec{A})\geq B_{\mathcal{C}}(\vec{A}).\label{eqn:boundingplane}
\end{align}
\end{theorem}

This is indeed our proof approach, however before giving the rather long proof for the general case, we first prove a few rate combinations for $K=3$, which illustrates the basic techniques as well as facilitates better understanding. Though the proof of the case for $K=3$ can also be found in \cite{GrokopTse:08}, our proof given here is different and in fact more structured, which is geared toward the general case. After this example, a few necessary tools and intermediate results are provided, and finally we give the converse proof of Theorem \ref{theorem:maintheorem}.

\begin{figure}[tb]
\begin{centering}
\includegraphics[width=8.5cm]{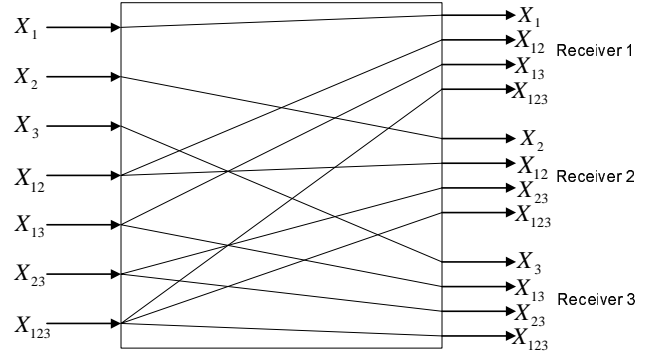}
\caption{The deterministic broadcast channel $K=3$ from \cite{GrokopTse:08}.\label{fig:deterministic}}
\end{centering}
\end{figure}

\subsection{Bounding Two Rate Combinations for $K=3$}
We give an outline of the proof for the first two inequalities in the example given after Theorem \ref{theorem:maintheorem}. 

\begin{proof}
\begin{align}
&3nR^*_1+6nR^*_2+2nR^*_3\nonumber\\
&\geq \frac{2}{3}\sum_{i=1}^3H(\mathcal{X}^n_i)+\frac{1}{3}\sum_{i=1}^3H(\mathcal{X}^n_i|X^n_{123},W_{123})\notag\\
&\stackrel{(a)}{\geq} \frac{2}{3}\sum_{i=1}^3[H(\mathcal{X}^n_i|\mathcal{W}_i)+H(\mathcal{W}_i)]\notag\\
&\qquad+\frac{1}{3}\sum_{i=1}^3H(\mathcal{X}^n_i|X^n_{123},W_{123})-n\delta\notag\\
&\stackrel{(b)}{=}2nR_1+4nR_2+2nR_3+\frac{2}{3}\sum_{i=1}^3H(\mathcal{X}^n_i|\mathcal{W}_i)\notag\\
&\qquad+\frac{1}{3}\sum_{i=1}^3H(\mathcal{X}^n_i|W_{123})-H(X^n_{123}|W_{123})-n\delta\notag\\
&\stackrel{(c)}{\geq}2nR_1+4nR_2+2nR_3+\frac{2}{3}H(X^n_{123}|W_{123})\notag\\
&\qquad+\frac{1}{3}\sum_{i=1}^3[H(\mathcal{W}_i|W_{123})+H(\mathcal{X}^n_i|\mathcal{W}_i)]\nonumber\\
&\qquad-H(X^n_{123}|W_{123})-n\delta'\notag\\
&=3nR_1+6nR_2+2nR_3-n\delta'\notag\\
&\qquad+\left[\frac{1}{3}\sum_{i=1}^3H(\mathcal{X}^n_i|\mathcal{W}_i)-\frac{1}{3}H(X^n_{123}|W_{123})\right]\notag\\
&\stackrel{(d)}{\geq} 3nR_1+6nR_2+2nR_3-n\delta',
\end{align}
where (a) is by Fano's inequality, (b) is by adding and subtracting the same term, (c) is by applying Fano's inequality on the third term, and noticing that Lemma \ref{lemma:Kwaysubmodularity} together with the fact of the channel being discrete implies that, 
\begin{align}
&\sum_{i=1}^3H(\mathcal{X}^n_i|\mathcal{W}_i)\notag\\
&\geq H(\overline{\mathcal{X}}^n_1|\overline{\mathcal{W}}_1)+H(\overline{\mathcal{X}}^n_{12}|\overline{\mathcal{W}}_{2})
+H(X^n_{123}|W_{123})\notag\\
&\geq \max\{H(X^n_{123}|W_{123}),H(\overline{\mathcal{X}}^n_2|\overline{\mathcal{W}}_2)\}\label{eqn:example1}
\end{align}
and (d) is again by the inequalities in (\ref{eqn:example1}). This completes the proof for the first rate combination. For the second rate combination, we have 
\begin{align}
&6nR^*_1+6nR^*_2+3nR^*_3\nonumber\\
&\geq \sum_{i=1}^3H(X^n_i|\overline{\mathcal{X}}^n_2\overline{\mathcal{W}}_2)+\sum_{i=1}^3H(\mathcal{X}^n_i)\notag\\
&\stackrel{(a)}{\geq} \sum_{i=1}^3H(X^n_i|\overline{\mathcal{X}}^n_2\overline{\mathcal{W}}_2)+3nR_1+6nR_2+3nR_3\nonumber\\
&\qquad+H(\overline{\mathcal{X}}^n_2|\overline{\mathcal{W}}_2)-n\delta\notag\\
&=\sum_{i=1}^3H(X^n_i\overline{\mathcal{X}}^n_2|\overline{\mathcal{W}}_2)-3H(\overline{\mathcal{X}}^n_2|\overline{\mathcal{W}}_2)\notag\\
&\qquad+3nR_1+6nR_2+3nR_3+H(\overline{\mathcal{X}}^n_2|\overline{\mathcal{W}}_2)-n\delta\notag\\
&\geq 6nR_1+6nR_2+3nR_3-n\delta'\notag\\
&\qquad+\left[\sum_{i=1}^3H(X^n_i\overline{\mathcal{X}}^n_2|\mathcal{W}_i\overline{\mathcal{W}}_2)-2H(\overline{\mathcal{X}}^n_2|\overline{\mathcal{W}}_2)\right]\notag\\
&\stackrel{(b)}{\geq}6nR_1+6nR_2+3nR_3-n\delta',\label{eqn:examplesecond}
\end{align}
where (a) is because of (\ref{eqn:example1}), and in (b) we applied Lemma \ref{lemma:Kwaysubmodularity}, 
\begin{align}
\sum_{i=1}^3H(X^n_i\overline{\mathcal{X}}^n_2|\mathcal{W}_i\overline{\mathcal{W}}_2)&\geq H(\overline{X}^n_1|\overline{\mathcal{W}}_1)+2H(\overline{\mathcal{X}}^n_2|\overline{\mathcal{W}}_2),
\end{align}
and then omit the first term since the channel is discrete; the rest of the inequalities in (\ref{eqn:examplesecond}) are by Fano's inequality.
\end{proof}

This proof illustrates several main components of the proof for the general case. Firstly, the rate combination needs to be written as summations under appropriate proportions, secondly the $K$-way submodularity lemma needs to strategically used, and thirdly there are connections between different layers of messages and thus terms may be canceled among them. For the general $K$-user problem, the bounding becomes much more complicated, and we will rely on the optimal solution $B_{C^*}(\vec{A})$ to provide necessary structure and guidance.  

\subsection{Several Properties of $\phi_{i,j}$}

We begin with a few properties on the exchange rate $\phi_{i,j}$.
\begin{lemma}
\label{lemma:increasingijk}
For any integers $i,j,k$ such that $1\leq i<j<k\leq K$, we have $\phi_{i,j}\phi_{j,k}=\phi_{i,k}$.
\end{lemma}
\begin{lemma}
\label{lemma:decreasingijk}
For any integers $i,j,k$ such that $1\leq i<j<k\leq K$, we have $\phi_{k,j}\phi_{j,i}=\phi_{k,i}$.
\end{lemma}

\begin{lemma}
\label{lemma:circleij}
For any integers $i,j$ such that $1\leq i<j \leq K$, we have $\phi_{i,j}\phi_{j,i}=i/j<1$.
\end{lemma}

\begin{lemma}
\label{lemma:twostep}
For any integers $i,j,k$, we have $\phi_{i,k}\geq \phi_{i,j}\phi_{j,k}$, with equality only when the sequence $(i,j,k)$ is monotonic. 
\end{lemma}

\begin{lemma}
\label{lemma:consequtivedown}
For any $k>j$, we have $(k+1)\binom{K-1}{k-1}\phi_{k+1,j}=k\binom{K-1}{k}\phi_{k,j}$.
\end{lemma}

\begin{lemma}
For any $i<j$, $\binom{K-1}{i-1}{\binom{K-1}{j-1}}^{-1}=\phi_{i,j}$.
\label{lemma:alternativephijk}
\end{lemma}

The above lemmas (particularly Lemma \ref{lemma:increasingijk}-\ref{lemma:twostep}) may be best understood as a currency exchange system where up-converting (or down-converting) many times results in the same final exchange rate as a single step conversion, but up-converting mixed with down-converting to the original currency results in a loss. The proofs of these lemmas are given in Appendix \ref{appendix:lemmas}.

\subsection{Extremal Solutions and the Effective Rate Set}
To prove the converse part of Theorem \ref{theorem:boundingplane}, we proceed in two steps: first we identify some special optimal solutions for the maximization problem (\ref{eqn:maximizingBstar}) with certain desired properties, then show that $B_{\mathcal{C}^*}(\vec{A})$ is an upper bound to  the quantity $B_{\mathcal{C}}(\vec{A})$. In this subsection we discuss the first step.

%The first part is essentially an analysis based on duality, however we will try to avoid the multipliers whenever possible.

%The following lemma gives a general property of the optimal solution for the maximization problem (\ref{eqn:maximizingBstar}).
%\begin{lemma}
%\label{lemma:optimalgeneral}
%For any distinct $j$ and $j^*$, if there exist $i$ and $i^*$ such that
%if $r_{i,j}>0$ and $r_{i^*,j^*}>0$, then it must to true that
%\begin{align}
%A_{j^*}\phi_{j,j^*}\leq A_j,\quad\mbox{and}\quad A_{j}\phi_{j^*,j}\leq A_{j^*}.
%\end{align}
%\end{lemma}

\begin{definition}
A non-negative setting of $r_{i,j}$ satisfying (\ref{eqn:condition1}) is called extremal if the following conditions hold (\textsf{i}) For each $i=1,2,\ldots,K$, there exists a unique $j\in \mathcal{I}_K$ such that $r_{i,j}=R^*_i$ and $r_{i,k}=0$ for $k\neq j$. (\textsf{ii}) If $r_{i,j}=R^*_i>0$, then $r_{j,j}=R^*_j$. (\textsf{iii}) If $r_{i,j}=R^*_i>0$, then for any $k$ such that $\max(i,j)>k>\min(i,j)$, $r_{k,j}=R^*_k$.
\end{definition}

\begin{lemma}
\label{lemma:extremal}
The solutions to the maximization problem (\ref{eqn:maximizingBstar}) include one that is extremal.
\end{lemma}

The lemma is intuitively true since a linear optimization problem has an optimal solution at its corner point. The concept of extremal solution makes the definition of corner point in the problem context more precise. A proof is given in Appendix \ref{appendix:lemmas}.

\begin{definition}
In an optimal extremal solution, the effective rate set is defined as $\mathcal{E}\triangleq\{i\in \mathcal{I}_K:r_{j,i}>0\mbox{ for some }j\}$. The elements of $\mathcal{E}$ in an increasing order are denoted as $e_1,e_2,\ldots,e_{|\mathcal{E}|}$.
\end{definition}

Lemma \ref{lemma:extremal} implies there exists a specific structure of rate exchange in the optimal extremal solutions.
\begin{lemma}
\label{lemma:quantization}
For an optimal extremal solution:
\begin{itemize}
\item There exist a partition of the sequence $1,2,\ldots,K$, labeled as $\mathcal{S}_1,\mathcal{S}_2,\ldots,\mathcal{S}_{|\mathcal{E}|}$, each consisting a consecutive sequence of integers, and $e_i\in \mathcal{S}_i$.
\item For $k\in \mathcal{S}_i$, we have $r_{k,e_i}=R^*_k$.
\end{itemize}
\end{lemma}

This structure is analogous to scalar quantization to some extent, as illustrated in Fig. \ref{fig:quantization}.

\begin{figure}[tb]
\begin{centering}
\includegraphics[width=8.5cm]{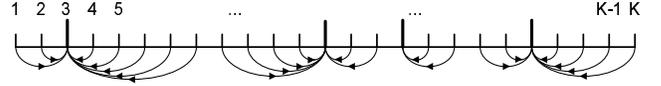}
\caption{An illustration of the optimal extremal solution structure. The longer and bolder marks give the set $\mathcal{E}$.\label{fig:quantization}}
\end{centering}
\end{figure}

\subsection{Proof of the Converse Part of Theorem \ref{theorem:boundingplane}}
%Now we are ready to prove the converse part of Theorem \ref{theorem:boundingplane}.
\begin{proof}%[Proof of Theorem \ref{theorem:boundingplane}]
For a fixed vector $\vec{A}$, let $\{\hat{r}_{i,j}\}$ be an optimal extremal solution for the maximization problem (\ref{eqn:maximizingBstar}), and let $\mathcal{E}$ be its effective rate set and let  $\mathcal{S}_1,\mathcal{S}_2,\ldots,\mathcal{S}_{|\mathcal{E}|}$ be the partition sets; for convenience, denote the smallest element in the set $\mathcal{S}_i$ as $l_i$ and the largest element as $u_i$. Assuming a sequence of length-$n$ codes is given with diminishing error probability. 
Let $\mathcal{X}_i$ and $\overline{\mathcal{X}}_i$ be defined similarly as $\mathcal{W}_i$ and $\overline{\mathcal{W}}_i$. The proof consists of two layers of inductions. We start from the inner layer, and then put the pieces together in the outer layer.

\begin{figure*}[b]
\normalsize
\newcounter{MYtempeqncnt}
\addtocounter{equation}{4}
\setcounter{MYtempeqncnt}{\value{equation}}
\hrulefill
\begin{align}
L_k\geq& \sum_{j=e_{k}}^{u_k-1}a_{k,j} \sum_{i=1}^KH(\mathcal{X}^n_i|\overline{\mathcal{X}}^n_{j+1}\overline{\mathcal{W}}_{j+1})+nK\left(\phi_{u_k,e_{k}}-\frac{A_{e_{k+1}}\phi_{u_k,e_{k+1}}}{A_{e_{k}}}\right)R_{u_k}\notag\\
&+a_{k,u_k}\sum_{i=1}^KH(\mathcal{W}_i|\overline{\mathcal{W}}_{u_k})+a_{k,u_k}\sum_{i=1}^{u_k}H(\overline{\mathcal{X}}^n_{i}|\overline{\mathcal{W}}_{i})-u_ka_{k,u_k}H(\overline{\mathcal{X}}^n_{u_k+1}|\overline{\mathcal{W}}_{u_k+1}).\label{eqn:Lkfirst}
\end{align}
\end{figure*}
\setcounter{equation}{\value{MYtempeqncnt}}
\addtocounter{equation}{-4}

Define the following quantity for $k=1,2,\ldots,|\mathcal{E}|$, for which lower bounds will be derived
\begin{align}
L_k&\triangleq\sum_{j=e_{k}}^{u_k-1}a_{k,j} \sum_{i=1}^KH(\mathcal{X}^n_i|\overline{\mathcal{X}}^n_{j+1}\overline{\mathcal{W}}_{j+1})\notag\\
&\qquad\qquad\qquad+a_{k,u_k}\sum_{i=1}^KH(\mathcal{X}^n_i|\overline{\mathcal{X}}^n_{u_k+1}\overline{\mathcal{W}}_{u_k+1}),
\end{align} 
where
\begin{align}
a_{k,j}&\triangleq \frac{\phi_{j,e_{k}}}{\binom{K-1}{j-1}}-\frac{\phi_{j+1,e_{k}}}{\binom{K-1}{j}},\,j=e_k,\ldots,u_k-1\notag\\
a_{k,u_k}&\triangleq\frac{\phi_{u_k,e_{k}}}{\binom{K-1}{u_k-1}}-\frac{A_{e_{k+1}}\phi_{u_k,e_{k+1}}}{A_{e_{k}}\binom{K-1}{u_k-1}},
\end{align}
and for convenience we have defined $A_{e_{|\mathcal{E}|+1}}\triangleq 0$ and $\phi_{j,e_{|\mathcal{E}|+1}}\triangleq 0$. Note that all the coefficients in front of the entropy functions are non-negative: those in the first summation are straightforward to verify by using the definition of $\phi_{i,j}$, and for the last term we only need to observe that $A_{e_{k}}\phi_{u_k,e_{k}}\geq A_{e_{k+1}}\phi_{u_k,e_{k+1}}$ by the optimality of the extremal solution. For convenience let us also define
\begin{align}
b_{k,j}\triangleq\frac{\phi_{j,e_{k}}}{\binom{K-1}{j-1}}-\frac{A_{e_{k+1}}\phi_{u_k,e_{k+1}}}{A_{e_{k}}\binom{K-1}{u_k-1}}=\sum_{i=j}^{u_k}a_{k,i},\,j=e_k,\ldots,u_k,
\end{align}
which are clearly non-negative quantities. We are interested in these quantities $L_k$'s because they are directly related with the rate combination being considered, as we shall see shortly.

We start by writing the following
\begin{align}
&\sum_{i=1}^KH(\mathcal{X}^n_i|\overline{\mathcal{X}}^n_{j+1}\overline{\mathcal{W}}_{j+1})\notag\\
&=\sum_{i=1}^KH(\mathcal{X}^n_i|\overline{\mathcal{X}}^n_{j+1}\overline{\mathcal{W}}_{j+1})+KH(\overline{\mathcal{X}}^n_{j+1}|\overline{\mathcal{W}}_{j+1})\notag\\
&\qquad\qquad\qquad\qquad-KH(\overline{\mathcal{X}}^n_{j+1}|\overline{\mathcal{W}}_{j+1})\notag\\
&=\sum_{i=1}^KH(\mathcal{X}^n_i,\overline{\mathcal{X}}^n_{j+1}|\overline{\mathcal{W}}_{j+1})-KH(\overline{\mathcal{X}}^n_{j+1}|\overline{\mathcal{W}}_{j+1})\notag\\
&\stackrel{(a)}{\geq}\sum_{i=1}^KH(\mathcal{W}_i|\overline{\mathcal{W}}_{j+1})+\sum_{i=1}^KH(\mathcal{X}^n_i,\overline{\mathcal{X}}^n_{j+1}|\overline{\mathcal{W}}_{j+1},\mathcal{W}_i)\notag\\
&\qquad\qquad\qquad\qquad-KH(\overline{\mathcal{X}}^n_{j+1}|\overline{\mathcal{W}}_{j+1})-n\delta\notag\\
&\stackrel{(b)}{\geq} \sum_{i=1}^KH(\mathcal{W}_i|\overline{\mathcal{W}}_{j+1})+\sum_{i=1}^jH(\overline{\mathcal{X}}^n_{i}|\overline{\mathcal{W}}_{i})\notag\\
&\qquad\qquad\qquad\qquad-jH(\overline{\mathcal{X}}^n_{j+1}|\overline{\mathcal{W}}_{j+1})-n\delta,\label{eqn:decomposedLk}
\end{align}
where (a) is by Fano's inequality, and (b) is by applying Lemma \ref{lemma:Kwaysubmodularity} on the second term. For notational simplicity, we shall ignore the small quantity $\delta$ in the sequel.

\begin{figure*}[tb]
\normalsize
\addtocounter{equation}{1}
\setcounter{MYtempeqncnt}{\value{equation}}
\begin{align}
L_k\geq& \sum_{j=e_{k}}^{m}a_{k,j} \sum_{i=1}^KH(\mathcal{X}^n_i|\overline{\mathcal{X}}^n_{j+1}\overline{\mathcal{W}}_{j+1})+nK\sum_{j=m+1}^{u_k}\left(\phi_{j,e_{k}}-\frac{A_{e_{k+1}}\phi_{j,e_{k+1}}}{A_{e_{k}}}\right)R_{j}+b_{k,m+1}\sum_{i=1}^KH(\mathcal{W}_i|\overline{\mathcal{W}}_{m+1})\notag\\
&+b_{k,m+1}\sum_{i=1}^{m+1}H(\overline{\mathcal{X}}^n_{i}|\overline{\mathcal{W}}_{i})-u_ka_{k,u_k}H(\overline{\mathcal{X}}^n_{u_k+1}|\overline{\mathcal{W}}_{u_k+1})-\frac{A_{e_{k+1}}\phi_{u_k,e_{k+1}}}{A_{e_{k}}\binom{K-1}{u_k-1}}\sum_{i=m+1}^{u_k-1}H(\overline{\mathcal{X}}^n_{i+1}|\overline{\mathcal{W}}_{i+1}).\label{eqn:Lkgenerala}
\end{align}
\hrulefill
\end{figure*}
\setcounter{equation}{\value{MYtempeqncnt}}
\addtocounter{equation}{1}
\begin{figure*}[tb]
\normalsize
\setcounter{MYtempeqncnt}{\value{equation}}
\begin{align}
L_k\geq& \sum_{j=e_{k}}^{m^*-1}a_{k,j} \sum_{i=1}^KH(\mathcal{X}^n_i|\overline{\mathcal{X}}^n_{j+1}\overline{\mathcal{W}}_{j+1})+a_{k,m^*}\sum_{i=1}^KH(\mathcal{W}_i|\overline{\mathcal{W}}_{m^*+1})+a_{k,m^*}\sum_{i=1}^{m^*}H(\overline{\mathcal{X}}^n_{i}|\overline{\mathcal{W}}_{i})\notag\\
&-a_{k,m^*}m^*H(\overline{\mathcal{X}}^n_{m^*+1}|\overline{\mathcal{W}}_{m^*+1})+nK\sum_{j=m^*+1}^{uk}\left(\phi_{j,e_{k}}-\frac{A_{e_{k+1}}\phi_{j,e_{k+1}}}{A_{e_{k}}}\right)R_{j}+b_{k,m^*+1}\sum_{i=1}^KH(\mathcal{W}_i|\overline{\mathcal{W}}_{m^*+1})\notag\\
&+b_{k,m^*+1}\sum_{i=1}^{m^*+1}H(\overline{\mathcal{X}}^n_{i}|\overline{\mathcal{W}}_{i})-u_ka_{k,u_k}H(\overline{\mathcal{X}}^n_{u_k+1}|\overline{\mathcal{W}}_{u_k+1})-\frac{A_{e_{k+1}}\phi_{u_k,e_{k+1}}}{A_{e_{k}}\binom{K-1}{u_k-1}}\sum_{i=m^*+1}^{u_k-1}H(\overline{\mathcal{X}}^n_{i+1}|\overline{\mathcal{W}}_{i+1}).\label{eqn:Lkinduction1}
\end{align}
\hrulefill
\end{figure*}
\setcounter{equation}{\value{MYtempeqncnt}}

\begin{figure*}[tb]
\normalsize
\addtocounter{equation}{5}
\setcounter{MYtempeqncnt}{\value{equation}}
\begin{align}
L_k\geq& \sum_{j=e_{k}}^{m^*-1}a_{k,j} 
\sum_{i=1}^KH(\mathcal{X}^n_i|\overline{\mathcal{X}}^n_{j+1}\overline{\mathcal{W}}_{j+1})+nK\sum_{j=m^*}^{u_k}\left(\phi_{j,e_{k}}-\frac{A_{e_{k+1}}\phi_{j,e_{k+1}}}{A_{e_{k}}}\right)R_{j}+b_{k,m^*}\sum_{i=1}^KH(\mathcal{W}_i|\overline{\mathcal{W}}_{m^*})\notag\\
&+b_{k,m^*}\sum_{i=1}^{m^*}H(\overline{\mathcal{X}}^n_{i}|\overline{\mathcal{W}}_{i})-u_ka_{k,u_k}H(\overline{\mathcal{X}}^n_{u_k+1}|\overline{\mathcal{W}}_{u_k+1})-\frac{A_{e_{k+1}}\phi_{u_k,e_{k+1}}}{A_{e_{k}}\binom{K-1}{u_k-1}}\sum_{i=m^*}^{u_k-1}H(\overline{\mathcal{X}}^n_{i+1}|\overline{\mathcal{W}}_{i+1}),\label{eqn:Lkinduction4}
\end{align}
\hrulefill
\end{figure*}
\setcounter{equation}{\value{MYtempeqncnt}}

\begin{figure*}[tb]
\normalsize
\addtocounter{equation}{1}
\setcounter{MYtempeqncnt}{\value{equation}}
\begin{align}
L_k\geq& nK\sum_{j=e_k}^{u_k}\left(\phi_{j,e_{k}}-\frac{A_{e_{k+1}}\phi_{j,e_{k+1}}}{A_{e_{k}}}\right)R_{j}+b_{k,e_k}\sum_{i=1}^KH(\mathcal{W}_i|\overline{\mathcal{W}}_{e_k})+b_{k,e_k}\sum_{i=1}^{e_k}H(\overline{\mathcal{X}}^n_{i}|\overline{\mathcal{W}}_{i})\notag\\
&-u_ka_{k,u_k}H(\overline{\mathcal{X}}^n_{u_k+1}|\overline{\mathcal{W}}_{u_k+1})-\frac{A_{e_{k+1}}\phi_{u_k,e_{k+1}}}{A_{e_{k}}\binom{K-1}{u_k-1}}\sum_{i=e_k}^{u_k-1}H(\overline{\mathcal{X}}^n_{i+1}|\overline{\mathcal{W}}_{i+1}).
\label{eqn:newlabel1}
\end{align}
\hrulefill
\end{figure*}
\setcounter{equation}{\value{MYtempeqncnt}}
\addtocounter{equation}{-5}

Slightly further expanding the first term in (\ref{eqn:decomposedLk}) and substituting it in $L_k$ give us (\ref{eqn:Lkfirst}). 
More generally, we claim that for $m$ such that $u_k-1\geq m\geq e_k-1$, (\ref{eqn:Lkgenerala}) holds, which we prove by induction. Clearly it holds for $m=u_k-1$ since it is exactly (\ref{eqn:Lkfirst}) in this case. Suppose it holds for $m=m^*$, we shall prove it also holds for $m=m^*-1$. Putting
 (\ref{eqn:decomposedLk}) into (\ref{eqn:Lkgenerala}), we have (\ref{eqn:Lkinduction1}) given on the next page. In order to simplify (\ref{eqn:Lkinduction1}), first notice that $a_{k,m^*}+b_{k,m^*+1}=b_{k,m^*}$, and 
\begin{align}
b_{k,m^*}\binom{K-1}{m^*-1}&=\phi_{m^*,e_{k}}-\frac{A_{e_{k+1}}\phi_{u_k,e_{k+1}}\binom{K-1}{m^*-1}}{A_{e_{k}}\binom{K-1}{u_k-1}}\notag\\
&\stackrel{(a)}{=}\phi_{m^*,e_{k}}-\frac{A_{e_{k+1}}\phi_{u_k,e_{k+1}}\phi_{m^*,u_k}}{A_{e_{k}}}\notag\\
&\stackrel{(b)}{=}\phi_{m^*,e_{k}}-\frac{A_{e_{k+1}}\phi_{m^*,e_{k+1}}}{A_{e_{k}}},\label{eqn:difference1}
\end{align}
where (a) is by Lemma \ref{lemma:alternativephijk} and (b) is by Lemma \ref{lemma:increasingijk}. It follows that
\begin{align}
&a_{k,m^*}\sum_{i=1}^KH(\mathcal{W}_i|\overline{\mathcal{W}}_{m^*+1})+b_{k,m^*+1}\sum_{i=1}^KH(\mathcal{W}_i|\overline{\mathcal{W}}_{m^*+1})\notag\\
&=b_{k,m^*}\sum_{i=1}^KH(\mathcal{W}_i|\overline{\mathcal{W}}_{m^*})+nKb_{k,m^*}\binom{K-1}{m^*-1}R_{m^*},\notag\\
&=b_{k,m^*}\sum_{i=1}^KH(\mathcal{W}_i|\overline{\mathcal{W}}_{m^*})\notag\\
&\qquad\qquad+nK\left(\phi_{m^*,e_{k}}-\frac{A_{e_{k+1}}\phi_{m^*,e_{k+1}}}{A_{e_{k}}}\right)R_{m^*}.
\label{eqn:Lkinduction2}
\end{align}
Furthermore, notice that
\begin{align}
&a_{k,m^*}\sum_{i=1}^{m^*}H(\overline{\mathcal{X}}^n_{i}|\overline{\mathcal{W}}_{i})-a_{k,m^*}m^*H(\overline{\mathcal{X}}^n_{m^*+1}|\overline{\mathcal{W}}_{m^*+1})\notag\\
&\qquad\qquad+b_{k,m^*+1}\sum_{i=1}^{m^*+1}H(\overline{\mathcal{X}}^n_{i}|\overline{\mathcal{W}}_{i})\notag\\
&=b_{k,m^*}\sum_{i=1}^{m^*}H(\overline{\mathcal{X}}^n_{i}|\overline{\mathcal{W}}_{i})\notag\\
&\qquad\qquad+(b_{k,m^*+1}-a_{k,m^*}m^*)H(\overline{\mathcal{X}}^n_{m^*+1}|\overline{\mathcal{W}}_{m^*+1})\notag\\
&=b_{k,m^*}\sum_{i=1}^{m^*}H(\overline{\mathcal{X}}^n_{i}|\overline{\mathcal{W}}_{i})\notag\\
&\qquad\qquad-\frac{A_{e_{k+1}}\phi_{u_k,e_{k+1}}}{A_{e_{k}}\binom{K-1}{u_k-1}}H(\overline{\mathcal{X}}^n_{m^*+1}|\overline{\mathcal{W}}_{m^*+1}),\label{eqn:Lkinduction3}
\end{align}
where the last step is due to
\begin{align}
&b_{k,m^*+1}-a_{k,m^*}m^*\notag\\
&=(m^*+1)\frac{\phi_{m^*+1,e_{k}}}{\binom{K-1}{m^*}}-\frac{A_{e_{k+1}}\phi_{u_k,e_{k+1}}}{A_{e_{k}}\binom{K-1}{u_k-1}}-m^*\frac{\phi_{m^*,e_{k}}}{\binom{K-1}{m^*-1}}\notag\\
&=-\frac{A_{e_{k+1}}\phi_{u_k,e_{k+1}}}{A_{e_{k}}\binom{K-1}{u_k-1}},
\end{align}
where the last equality is by Lemma \ref{lemma:consequtivedown}.
Combining (\ref{eqn:Lkinduction1}), (\ref{eqn:Lkinduction2}) and (\ref{eqn:Lkinduction3}), we have (\ref{eqn:Lkinduction4}), proving that the claim (\ref{eqn:Lkgenerala}) is indeed true.

Letting $m=e_k-1$, we can write (\ref{eqn:newlabel1}) on this page. By breaking the second term as given in (\ref{eqn:difference2}),
\begin{figure*}[tb]
\normalsize
\addtocounter{equation}{2}
\setcounter{MYtempeqncnt}{\value{equation}}
\begin{align}
b_{k,e_k}\sum_{i=1}^KH(\mathcal{W}_i|\overline{\mathcal{W}}_{e_k})&=b_{k,e_k}nK\sum_{j=l_k}^{e_k-1}\binom{K-1}{j-1}R_j+b_{k,e_k}\sum_{i=1}^KH(\mathcal{W}_i|\overline{\mathcal{W}}_{l_k})\notag\\
&=nK\sum_{j=l_k}^{e_k-1}\left(\frac{\binom{K-1}{j-1}}{\binom{K-1}{e_k-1}}-\frac{A_{e_{k+1}}\phi_{u_k,e_{k+1}}\binom{K-1}{j-1}}{A_{e_{k}}\binom{K-1}{u_k-1}}\right)R_j+b_{k,e_k}\sum_{i=1}^KH(\mathcal{W}_i|\overline{\mathcal{W}}_{l_k})\notag\\
&=nK\sum_{j=l_k}^{e_k-1}\left(\phi_{j,e_{k}}-\frac{A_{e_{k+1}}\phi_{j,e_{k+1}}}{A_{e_k}}\right)R_j+b_{k,e_k}\sum_{i=1}^KH(\mathcal{W}_i|\overline{\mathcal{W}}_{l_k}),\label{eqn:difference2}
\end{align}
\hrulefill
\end{figure*}
\begin{figure*}[tb]
\normalsize
\setcounter{MYtempeqncnt}{\value{equation}}
\addtocounter{MYtempeqncnt}{1}
\begin{align}
L_k\geq& nK\sum_{j=l_k}^{u_k}\left(\phi_{j,e_{k}}-\frac{A_{e_{k+1}}\phi_{j,e_{k+1}}}{A_{e_{k}}}\right)R_{j}+b_{k,e_k}\sum_{i=1}^KH(\mathcal{W}_i|\overline{\mathcal{W}}_{l_k})+b_{k,e_k}\sum_{i=1}^{l_k}H(\overline{\mathcal{X}}^n_{i}|\overline{\mathcal{W}}_{i})\notag\\
&-u_ka_{k,u_k}H(\overline{\mathcal{X}}^n_{u_k+1}|\overline{\mathcal{W}}_{u_k+1})-\frac{A_{e_{k+1}}\phi_{u_k,e_{k+1}}}{A_{e_{k}}\binom{K-1}{u_k-1}}\sum_{i=e_k}^{u_k-1}H(\overline{\mathcal{X}}^n_{i+1}|\overline{\mathcal{W}}_{i+1}).\label{eqn:Lkgeneral}
\end{align}
\hrulefill
\end{figure*}
where in the last step we apply Lemma \ref{lemma:alternativephijk} and Lemma \ref{lemma:increasingijk}, 
and noticing that for the third term
\begin{align*}
b_{k,e_k}\sum_{i=1}^{e_k}H(\overline{\mathcal{X}}^n_{i}|\overline{\mathcal{W}}_{i})\geq b_{k,e_k}\sum_{i=1}^{l_k}H(\overline{\mathcal{X}}^n_{i}|\overline{\mathcal{W}}_{i}),
\end{align*}
implied by the discrete nature of the channel, we can further write (\ref{eqn:Lkgeneral}) on the next page.

This concludes the inner layer induction, and next we turn to the outer layer. First notice that the optimality of extremal solution and Lemma \ref{lemma:quantization} together imply that
\begin{align}
B_{C^*}(\vec{A})&=\sum_{i=1}^{|\mathcal{E}|}A_{e_i}\sum_{j\in \mathcal{S}_i}\phi_{j,e_i}R^*_j.
\end{align}

\begin{figure*}[tb]
\normalsize
\setcounter{MYtempeqncnt}{\value{equation}}
\addtocounter{MYtempeqncnt}{1}
\begin{align}
B_{C^*}(\vec{A})\geq&
\sum_{i=1}^{|\mathcal{E}|-1}\sum_{j\in \mathcal{S}_i}\left[A_{e_i}\phi_{j,e_i}-A_{e_{|\mathcal{E}|}}\phi_{j,e_{|\mathcal{E}|}}\right]R^*_j+\frac{A_{e_{|\mathcal{E}|}}}{nK}\sum_{j=e_{|\mathcal{E}|}}^K \left(\frac{\phi_{j,e_{|\mathcal{E}|}}}{\binom{K-1}{j-1}}-\frac{\phi_{j+1,e_{|\mathcal{E}|}}}{\binom{K-1}{j}}\right)\sum_{i=1}^KH(\mathcal{X}^n_i|\overline{\mathcal{X}}^n_{j+1}\overline{\mathcal{W}}_{j+1})\notag\\
=&\sum_{i=1}^{|\mathcal{E}|-1}\sum_{j\in \mathcal{S}_i}\left[A_{e_i}\phi_{j,e_i}-A_{e_{|\mathcal{E}|}}\phi_{j,e_{|\mathcal{E}|}}\right]R^*_j+\frac{A_{e_{|\mathcal{E}|}}}{nK}L_{|\mathcal{E}|},\label{eqn:expansion}
\end{align}
\hrulefill
\end{figure*}
We first write (\ref{eqn:expansion}) where the inequality can be justified as follows. 
Observe that in the second summation, for any $k>e_{|\mathcal{E}|}$, the random variables $X_{\mathcal{A}}$ with $|\mathcal{A}|=k$ appear only in the last $K-k+1$ terms in the outer summation. Each inner summation has a total of $K\binom{K-1}{j-1}$ such terms, which implies such random variables are counted a total of $A_{e_{|\mathcal{E}|}}K\phi_{j,e_{|\mathcal{E}|}}$ times. Thus by the cardinality of the alphabets, the normalized entropy is upper bounded by $R^*_{k}$. Through a similar argument, it is not difficult to verify that for $k\leq e_{|\mathcal{E}|}$, all the terms are accounted for. Furthermore, notice that by the optimality of the extremal solution, for any $j\in\mathcal{S}_i$, we have $A_{e_i}\phi_{j,e_i}\geq A_{e_{|\mathcal{E}|}}\phi_{j,e_{|\mathcal{E}|}}$, and thus the first summation is non-negative.

\begin{figure*}[tb]
\normalsize
\setcounter{MYtempeqncnt}{\value{equation}}
\addtocounter{MYtempeqncnt}{1}
\begin{align}
B_{C^*}(\vec{A})\geq&\sum_{i=1}^{|\mathcal{E}|-1}\sum_{j\in \mathcal{S}_i}\left[A_{e_i}\phi_{j,e_i}-A_{e_{|\mathcal{E}|}}\phi_{j,e_{|\mathcal{E}|}}\right]R^*_j+A_{e_{|\mathcal{E}|}}\sum_{j=l_{|\mathcal{E}|}}^{u_{|\mathcal{E}|}}\left(\phi_{j,e_{|\mathcal{E}|}}-\frac{A_{e_{|\mathcal{E}|+1}}\phi_{j,e_{|\mathcal{E}|+1}}}{A_{e_{|\mathcal{E}|}}}\right)R_{j}\notag\\
&+\frac{A_{e_{|\mathcal{E}|}}}{nK}b_{{|\mathcal{E}|},e_{|\mathcal{E}|}}\left(\sum_{i=1}^KH(\mathcal{W}_i|\overline{\mathcal{W}}_{l_{|\mathcal{E}|}})+\sum_{i=1}^{l_{|\mathcal{E}|}}H(\overline{\mathcal{X}}^n_{i}|\overline{\mathcal{W}}_{i})\right)-\frac{A_{e_{|\mathcal{E}|}}}{nK}u_{|\mathcal{E}|}a_{{|\mathcal{E}|},u_{|\mathcal{E}|}}H(\overline{\mathcal{X}}^n_{u_{|\mathcal{E}|}+1}|\overline{\mathcal{W}}_{u_{|\mathcal{E}|}+1})\notag\\
&-\frac{A_{e_{|\mathcal{E}|}}}{nK}\frac{A_{e_{{|\mathcal{E}|}+1}}\phi_{u_{|\mathcal{E}|},e_{{|\mathcal{E}|}+1}}}{A_{e_{|\mathcal{E}|}}\binom{K-1}{u_{|\mathcal{E}|}-1}}\sum_{i=e_{|\mathcal{E}|}}^{u_{|\mathcal{E}|}-1}H(\overline{\mathcal{X}}^n_{i+1}|\overline{\mathcal{W}}_{i+1})\notag\\
=&\sum_{i=1}^{|\mathcal{E}|-1}\sum_{j\in \mathcal{S}_i}\left[A_{e_i}\phi_{j,e_i}-A_{e_{|\mathcal{E}|}}\phi_{j,e_{|\mathcal{E}|}}\right]R^*_j+A_{e_{|\mathcal{E}|}}\sum_{j=l_{|\mathcal{E}|}}^K\phi_{j,e_{|\mathcal{E}|}}R_j+\frac{A_{e_{|\mathcal{E}|}}}{nK\binom{K-1}{e_{|\mathcal{E}|}-1}}\sum_{i=1}^KH(\mathcal{W}_i|\overline{\mathcal{W}}_{l_{|\mathcal{E}|}})\notag\\
&+
\frac{A_{e_{|\mathcal{E}|}}}{nK\binom{K-1}{e_{|\mathcal{E}|}-1}}\sum_{i=1}^{l_{|\mathcal{E}|}}H(\overline{\mathcal{X}}^n_{i}|\overline{\mathcal{W}}_{i})
,\label{eqn:outerinduction1}
\end{align}
\hrulefill
\end{figure*}
We next apply (\ref{eqn:Lkgeneral}) with $k=|\mathcal{E}|$ in (\ref{eqn:expansion}), and write (\ref{eqn:outerinduction1}) on the next page, because we have $u_{|\mathcal{E}|}=K$, $A_{e_{|\mathcal{E}|+1}}=0$ by definition, and 
\begin{align}
b_{|\mathcal{E}|,e_{|\mathcal{E}|}}=\frac{\phi_{e_{|\mathcal{E}|},e_{|\mathcal{E}|}}}{\binom{K-1}{e_{|\mathcal{E}|
}-1}}=\frac{1}{\binom{K-1}{e_{|\mathcal{E}|}-1}}.
\end{align}

More generally, we claim for $k=0,1,\ldots,|\mathcal{E}|-1$, the following inequality holds
\begin{align}
&B_{C^*}(\vec{A})\geq \sum_{i=1}^{k}\sum_{j\in \mathcal{S}_i}\left[A_{e_i}\phi_{j,e_i}-A_{e_{k+1}}\phi_{j,e_{k+1}}\right]R^*_j\notag\\
&\,+\sum_{i=k+1}^{|\mathcal{E}|}A_{e_i}\sum_{j=l_i}^{u_i}\phi_{j,e_{i}}R_j+\frac{A_{e_{k+1}}}{nK\binom{K-1}{e_{k+1}-1}}\sum_{i=1}^KH(\mathcal{W}_i|\overline{\mathcal{W}}_{l_{k+1}})\notag\\
&\qquad+
\frac{A_{e_{k+1}}}{nK\binom{K-1}{e_{k+1}-1}}\sum_{i=1}^{l_{k+1}}H(\overline{\mathcal{X}}^n_{i}|\overline{\mathcal{W}}_{i})\label{eqn:outergeneral}
\end{align}

\begin{figure*}[tb]
\normalsize
\addtocounter{equation}{4}
\setcounter{MYtempeqncnt}{\value{equation}}
\begin{align}
B_{C^*}(\vec{A})&\geq \sum_{i=1}^{k^*-1}\sum_{j\in \mathcal{S}_i}c_{j,k^*-1}R^*_j+A_{e_{k^*}}\sum_{j=l_{k^*}}^{u_{k^*}}\left(\phi_{j,e_{k^*}}-\frac{A_{e_{k^*+1}}\phi_{j,e_{k^*+1}}}{A_{e_{k^*}}}\right)R_{j}\notag\\
&+\frac{A_{e_{k^*}}}{nK}b_{k^*,e_{k^*}}\left(\sum_{i=1}^KH(\mathcal{W}_i|\overline{\mathcal{W}}_{l_{k^*}})+\sum_{i=1}^{l_{k^*}}H(\overline{\mathcal{X}}^n_{i}|\overline{\mathcal{W}}_{i})\right)-\frac{A_{e_{k^*}}}{nK}u_{k^*}a_{k^*,u_{k^*}}H(\overline{\mathcal{X}}^n_{u_{k^*}+1}|\overline{\mathcal{W}}_{u_{k^*}+1})\notag\\
&-\frac{A_{e_{{k^*}+1}}\phi_{u_{k^*},e_{{k^*}+1}}}{nK\binom{K-1}{u_{k^*}-1}}\sum_{i=e_{k^*}}^{u_{k^*}-1}H(\overline{\mathcal{X}}^n_{i+1}|\overline{\mathcal{W}}_{i+1})+\sum_{i=k^*+1}^{|\mathcal{E}|}A_{e_i}\sum_{j=l_i}^{u_i}\phi_{j,e_{i}}R_j\notag\\
&+\frac{A_{e_{k^*+1}}}{nK\binom{K-1}{e_{k^*+1}-1}}\sum_{i=1}^KH(\mathcal{W}_i|\overline{\mathcal{W}}_{l_{k^*+1}})+\frac{A_{e_{k^*+1}}}{nK\binom{K-1}{e_{k^*+1}-1}}\sum_{i=1}^{l_{k^*+1}}H(\overline{\mathcal{X}}^n_{i}|\overline{\mathcal{W}}_{i}).\label{eqn:outerind1}
\end{align}
\hrulefill
\end{figure*}
\setcounter{equation}{\value{MYtempeqncnt}}
\begin{figure*}[tb]
\normalsize
\addtocounter{equation}{1}
\setcounter{MYtempeqncnt}{\value{equation}}
\begin{align}
&A_{e_{k^*}}\sum_{j=l_{k^*}}^{u_{k^*}}\left(\phi_{j,e_{k^*}}-\frac{A_{e_{k^*+1}}\phi_{j,e_{k^*+1}}}{A_{e_{k^*}}}\right)R_{j}+\frac{A_{e_{k^*}}}{nK}b_{k^*,e_{k^*}}\sum_{i=1}^KH(\mathcal{W}_i|\overline{\mathcal{W}}_{l_{k^*}})+\sum_{i=k^*+1}^{|\mathcal{E}|}A_{e_i}\sum_{j=l_i}^{u_i}\phi_{j,e_{i}}R_j\notag\\
&\qquad+\frac{A_{e_{k^*+1}}}{nK\binom{K-1}{e_{k^*+1}-1}}\sum_{i=1}^KH(\mathcal{W}_i|\overline{\mathcal{W}}_{l_{k^*+1}})\notag\\
&=\sum_{j=l_{k^*}}^{u_{k^*}}\left(A_{e_{k^*}}\phi_{j,e_{k^*}}-A_{e_{k^*+1}}\phi_{j,e_{k^*+1}}\right)R_{j}+\frac{A_{e_{k^*}}}{nK}b_{k^*,e_{k^*}}\sum_{i=1}^KH(\mathcal{W}_i|\overline{\mathcal{W}}_{l_{k^*}})+\sum_{i=k^*+1}^{|\mathcal{E}|}A_{e_i}\sum_{j=l_i}^{u_i}\phi_{j,e_{i}}R_j\notag\\
&\qquad+\frac{A_{e_{k^*+1}}}{nK\binom{K-1}{e_{k^*+1}-1}}\sum_{i=1}^KH(\mathcal{W}_i|\overline{\mathcal{W}}_{l_{k^*}})+\sum_{j=l_{k^*}}^{u_{k^*}}\frac{A_{e_{k^*+1}}}{\binom{K-1}{e_{k^*+1}-1}}\binom{K-1}{j-1}R_{j}\notag\\
&\stackrel{(a)}{=}\sum_{i=k^*}^{|\mathcal{E}|}A_{e_i}\sum_{j=l_i}^{u_i}\phi_{j,e_{i}}R_j+\left(\frac{A_{e_{k^*}}}{nK}b_{k^*,e_{k^*}}+\frac{A_{e_{k^*+1}}}{nK\binom{K-1}{e_{k^*+1}-1}}\right)\sum_{i=1}^KH(\mathcal{W}_i|\overline{\mathcal{W}}_{l_{k^*}})\notag\\
&\stackrel{(b)}{=}\sum_{i=k^*}^{|\mathcal{E}|}A_{e_i}\sum_{j=l_i}^{u_i}\phi_{j,e_{i}}R_j+\frac{A_{e_{k^*}}}{nK\binom{K-1}{e_{k^*}-1}}\sum_{i=1}^KH(\mathcal{W}_i|\overline{\mathcal{W}}_{l_{k^*}}),
\label{eqn:outerind2}
\end{align}
\hrulefill
\end{figure*}
\setcounter{equation}{\value{MYtempeqncnt}}%\addtocounter{equation}{1}
\addtocounter{equation}{-5}

We again take an induction approach to prove this claim. The claim is clearly true for $k=|\mathcal{E}|-1$. Now suppose (\ref{eqn:outergeneral}) is true for $k=k^*$, and we seek to show it is also true for $k=k^*-1$. 
For notational simplicity, let us define
\begin{align}
c_{j,k}=A_{e_i}\phi_{j,e_i}-A_{e_{k+1}}\phi_{j,e_{k+1}}.
\end{align}
We first prove the following inequality
\begin{align}
&\sum_{i=1}^{k^*}\sum_{j\in \mathcal{S}_i}c_{j,k^*}R^*_j\geq\sum_{i=1}^{k^*-1}\sum_{j\in \mathcal{S}_i}c_{j,k^*-1}R^*_j+\frac{A_{e_{k^*}}}{nK}L_{k^*}.\label{eqn:accountrate}
\end{align}
To do this, we need to count in the second term the number of appearance of random variables $X_{\mathcal{A}}$ for all $|\mathcal{A}|=m$, for all fixed $m$, such that $m\in \mathcal{I}_{u_{k^*}}$. This is similar to (\ref{eqn:expansion}), but slightly more involved. For $m$ such that $u_{k^*}\geq m>e_{k^*}$, it is easily seen that there are a total of $b_{k^*,m}K\binom{K-1}{m-1}$ such random variables in $L_{k^*}$, implying the following amount of $R^*_m$ is accounted for 
\begin{align}
&\frac{A_{e_{k^*}}}{K}b_{k^*,m}K\binom{K-1}{m-1}=A_{e_{k^*}}\phi_{m,e_{k^*}}-A_{e_{k^*+1}}\phi_{m,e_{k^*+1}},
\end{align} 
where we have used (\ref{eqn:difference1}). This indeed is the difference between the left hand side of (\ref{eqn:accountrate}) and the first term on the right hand side, in terms of $R^*_m$. For the case $m\leq e_{k^*}$, the following amount of $R^*_m$ is accounted for
\begin{align}
&\frac{A_{e_{k^*}}}{K}b_{k^*,e_{k^*}}K\binom{K-1}{m-1}=A_{e_{k^*}}\phi_{m,e^{k^*}}-A_{e_{k^*+1}}\phi_{m,e_{k^*+1}},
\end{align}
where we have used the derivation in (\ref{eqn:difference2}). This is again precisely the difference between the left hand side of (\ref{eqn:accountrate}) and the first term on the right hand side, in terms of $R^*_m$. Thus (\ref{eqn:accountrate}) is indeed true. 
%\addtocounter{equation}{1}

Now we proceed with the proof of (\ref{eqn:outergeneral}) through induction by assuming it holds for $k=k^*$, and write (\ref{eqn:outerind1}) on the top of this page by applying (\ref{eqn:Lkgeneral}). In order to simplify (\ref{eqn:outerind1}), similar terms need to be combined, for which  we write (\ref{eqn:outerind2}),
where in (a) we used Lemma \ref{lemma:alternativephijk}, and (b) is because
\begin{align}
\addtocounter{equation}{2}
&\frac{A_{e_{k^*}}}{nK}b_{e_{k^*},e_{k^*}}+\frac{A_{e_{k^*+1}}}{nK\binom{K-1}{e_{k^*+1}-1}}\notag\\
&=\frac{A_{e_{k^*}}}{nK\binom{K-1}{k^*-1}}-\frac{A_{e_{k^*+1}}\phi_{u_k,e_{k^*+1}}}{nK\binom{K-1}{u_{k^*}-1}}+\frac{A_{e_{k^*+1}}}{nK\binom{K-1}{e_{k^*+1}-1}}\notag\\
&=\frac{A_{e_{k^*}}}{nK\binom{K-1}{k^*-1}}-\frac{A_{e_{k^*+1}}}{nK\binom{K-1}{u_{k^*}-1}}\left(\phi_{u_k,e_{k^*+1}}-\frac{\binom{K-1}{u_{k^*}-1}}{\binom{K-1}{e_{k^*+1}-1}}\right)\notag\\
&=\frac{A_{e_{k^*}}}{nK\binom{K-1}{k^*-1}}, \label{eqn:summationmiddle}
\end{align}
where the last step is again by Lemma \ref{lemma:alternativephijk}.

\begin{figure*}[tb]
\normalsize
\setcounter{MYtempeqncnt}{\value{equation}}
\addtocounter{MYtempeqncnt}{1}
\begin{align}
&\frac{A_{e_{k^*}}}{nK}b_{k^*,e_{k^*}}\sum_{i=1}^{l_{k^*}}H(\overline{\mathcal{X}}^n_{i}|\overline{\mathcal{W}}_{i})-\frac{A_{e_{k^*}}}{nK}u_{k^*}a_{k^*,u_{k^*}}H(\overline{\mathcal{X}}^n_{u_{k^*}+1}|\overline{\mathcal{W}}_{u_{k^*}+1})\notag\\
&\qquad-\frac{A_{e_{{k^*}+1}}\phi_{u_{k^*},e_{{k^*}+1}}}{nK\binom{K-1}{u_{k^*}-1}}\sum_{i=e_{k^*}}^{u_{k^*}-1}H(\overline{\mathcal{X}}^n_{i+1}|\overline{\mathcal{W}}_{i+1})+\frac{A_{e_{k^*+1}}}{nK\binom{K-1}{e_{k^*+1}-1}}\sum_{i=1}^{l_{k^*+1}}H(\overline{\mathcal{X}}^n_{i}|\overline{\mathcal{W}}_{i})\notag\\
&\geq\frac{A_{e_{k^*}}}{nK\binom{K-1}{k^*-1}}\sum_{i=1}^{l_{k^*}}H(\overline{\mathcal{X}}^n_{i}|\overline{\mathcal{W}}_{i})+\left(\frac{A_{e_{k^*+1}}}{nK\binom{K-1}{e_{k^*+1}-1}}-\frac{A_{e_{k^*}}}{nK}u_{k^*}a_{k^*,u_{k^*}}\right)H(\overline{\mathcal{X}}^n_{u_{k^*}+1}|\overline{\mathcal{W}}_{u_{k^*}+1})\notag\\
&\qquad+\left(\frac{A_{e_{k^*+1}}}{nK\binom{K-1}{e_{k^*+1}-1}}-\frac{A_{e_{{k^*}+1}}\phi_{u_{k^*},e_{{k^*}+1}}}{nK\binom{K-1}{u_{k^*}-1}}\right)\sum_{i=e_{k^*}}^{u_{k^*}-1}H(\overline{\mathcal{X}}^n_{i+1}|\overline{\mathcal{W}}_{i+1}),
\label{eqn:outerind3}
\end{align}
\hrulefill
\end{figure*}
\setcounter{equation}{\value{MYtempeqncnt}}

Next consider the summation (\ref{eqn:outerind3}) where we have split the last term and combined it with the other terms, and used (\ref{eqn:summationmiddle}); moreover, some terms $H(\overline{\mathcal{X}}^n_{i}|\overline{\mathcal{W}}_{i})$ for $i=l_{k^*}+1,\ldots,e_{k^*}$ are ignored because they are non-negative by the discrete nature of the channel. Observe that for the last term in the right hand side of (\ref{eqn:outerind3})
\begin{align}
&\frac{A_{e_{k^*+1}}}{nK\binom{K-1}{e_{k^*+1}-1}}-\frac{A_{e_{{k^*}+1}}\phi_{u_{k^*},e_{{k^*}+1}}}{nK\binom{K-1}{u_{k^*}-1}}\notag\\
&=\frac{A_{e_{k^*+1}}\phi_{u_{k^*},e_{k^*+1}}-A_{e_{{k^*}+1}}\phi_{u_{k^*},e_{{k^*}+1}}}{nK\binom{K-1}{u_{k^*}-1}}=0.
\label{eqn:outerind4}
\end{align}
For the second term in the right hand side of (\ref{eqn:outerind3}), notice that $u_{k^*}+1\in \mathcal{S}_{k^*+1}$, thus by the optimality of the extremal solution, we have
\begin{align}
A_{e_{k^*+1}}\phi_{u_{k^*}+1,e_{k^*+1}}\geq A_{e_{k^*}}\phi_{u_{k^*}+1,e_{k^*}},\label{eqn:changepoint}
\end{align}
and thus (\ref{eqn:outerind5}) follows, where (a) is by Lemma \ref{lemma:circleij}, and the final inequality is by (\ref{eqn:changepoint}). 
Now combining (\ref{eqn:outerind1}), (\ref{eqn:outerind2}), (\ref{eqn:outerind3}), (\ref{eqn:outerind4}) and (\ref{eqn:outerind5}) completes the induction proof of (\ref{eqn:outergeneral}) for $k=k^*-1$. 
\begin{figure*}[tb]
\normalsize
\setcounter{MYtempeqncnt}{\value{equation}}
\addtocounter{MYtempeqncnt}{1}
\begin{align}
\frac{A_{e_{k^*+1}}}{nK\binom{K-1}{e_{k^*+1}-1}}-\frac{A_{e_{k^*}}}{nK}u_{k^*}a_{k^*,u_{k^*}}
&=\frac{A_{e_{k^*+1}}}{nK\binom{K-1}{e_{k^*+1}-1}}-\frac{A_{e_{k^*}}u_{k^*}}{nK}\left(\frac{\phi_{u_{k^*},e_{k^*}}}{\binom{K-1}{u_{k^*}-1}}-\frac{A_{e_{k^*+1}}\phi_{u_{k^*},e_{k^*+1}}}{A_{e_{k^*}}\binom{K-1}{u_{k^*}-1}}\right)\notag\\
&=\frac{(u_{k^*}+1)A_{e_{k^*+1}}\phi_{u_{k^*},e_{k^*+1}}-A_{e_{k^*}}u_{k^*}\phi_{u_{k^*},e_{k^*}}}{nK\binom{K-1}{u_{k^*}-1}}\notag\\
&\stackrel{(a)}{=}\frac{A_{e_{k^*+1}}\phi_{u_{k^*}+1,e_{k^*+1}}-A_{e_{k^*}}u_{k^*}\phi_{u_{k^*},e_{k^*}}\phi_{u_{k^*}+1,u_{k^*}}}{nK\binom{K-1}{u_{k^*}-1}\phi_{u_{k^*}+1,u_{k^*}}}\geq 0,\label{eqn:outerind5}
\end{align}
\hrulefill
\end{figure*}
\setcounter{equation}{\value{MYtempeqncnt}}

Writing (\ref{eqn:outergeneral}) for $k=0$, we have
\begin{align}
&B_{C^*}(\vec{A})\geq \sum_{i=1}^{|\mathcal{E}|}A_{e_i}\sum_{j=l_i}^{u_i}\phi_{j,e_{i}}R_j\notag\\
&\qquad+\frac{A_{e_{1}}}{nK\binom{K-1}{e_{1}-1}}\left(\sum_{i=1}^KH(\mathcal{W}_i|\overline{\mathcal{W}}_{l_{1}})+\sum_{i=1}^{l_{1}}H(\overline{\mathcal{X}}^n_{i}|\overline{\mathcal{W}}_{i})\right)\notag\\
&\qquad\geq \sum_{i=1}^{|\mathcal{E}|}A_{e_i}\sum_{j=l_i}^{u_i}\phi_{j,e_{i}}R_j,\label{eqn:final}
\end{align}
where the second inequality is because the first term in the parenthesis degenerates to zero, and the second is non-negative. Notice that for any $j\in \mathcal{S}_i$, by the optimality of the given extremal solution, $A_j\leq A_{e_{i}}\phi_{j,e_{i}}$, thus by the non-negativeness of rate $R_i$'s, we arrive at
\begin{align}
B_{C^*}(\vec{A})&\geq \sum_{i=1}^{|\mathcal{E}|}A_{e_i}\sum_{j=l_i}^{u_i}\phi_{j,e_{i}}R_j\geq \sum_{i=1}^{K}A_{i}R_i.
\end{align}
This completes the proof.
\end{proof}

For the case that $2^{R^*_i}$'s are not integers, we can instead consider a sequence of channels with memory, for which the alphabet sizes are $2^{nR^*_i}$, however, for each $n$ channel use, the channel erases $(n-1)$ of them. This channel is not a memoryless channel anymore, however, our definition is sufficiently general to include such a case, and the converse proof can be used without any change. 

\section{Conclusion}
We consider the latent capacity region of the symmetric broadcast problem, which gives the maximum implication region for a specific achievable rate vector. A complete characterization is provided, for which the converse proof relies on a deterministic channel model, and deriving upper bounds for any bounding plane of the rate region. The forward proof reveals an inherent connection between broadcast with common messages and erasure correction codes.

We believe the latent capacity region (or latent rate region) is a general concept, and can be applied to other problem. In \cite{GrokopTse:08}, the multiple access channel is also considered for two and three-user case. It is conceivable that the technique used in this work can be used to generalize their results for the multiple access channel. Another interesting case may be the interference channel, where the well-known Han-Kobayashi region \cite{HanKobayashi:81} is indeed the projection of a rate region for the coding problem with common messages. A careful analysis of the latent capacity region for the general interference channel may yield further insight into the problem.

\section*{Acknowledgment} 
The author wishes to thank the anonymous reviewers for their comments which help improve the presentation of this paper. 

\appendices

\section{Submodularity property of conditional entropy}
\label{append:Kwaysubmodularity}
\begin{lemma}\label{lemma:submodular}
Let $U_1,U_2,\ldots,U_N$ be a set of mutually independent random variables, and let $V_1,V_2,\ldots,V_N$ be $N$ random variables jointly distributed with them. Let $\mathcal{G}$ be a subset of $\mathcal{I}_N$, i.e., $\mathcal{G}\subseteq \mathcal{I}_N$. The conditional entropy function  $H_{V|U}(\mathcal{G})\triangleq H(V_i,i\in \mathcal{G}|U_i,i\in\mathcal{G})$ is a submodular function, i.e., for any  $\mathcal{G}_1, \mathcal{G}_2\subseteq \mathcal{I}_N$, 
\begin{align}
H_{V|U}(\mathcal{G}_1)+H_{V|U}(\mathcal{G}_2)\geq H_{V|U}(\mathcal{G}_1\cup \mathcal{G}_2)+H_{V|U}(\mathcal{G}_1\cap \mathcal{G}_2).\notag
\end{align}
\end{lemma}
\begin{proof}
Notice that
\begin{align}
&H_{V|U}(\mathcal{G}_1)+H_{V|U}(\mathcal{G}_2)\notag\\
&=H(V_i,U_i,i\in \mathcal{G}_1)+H(V_i,U_i,i\in \mathcal{G}_2)\notag\\
&\quad-H(U_i,i\in \mathcal{G}_1)-H(U_i,i\in \mathcal{G}_2),
\end{align}
and 
\begin{align}
&H_{V|U}(\mathcal{G}_1\cup \mathcal{G}_2)+H_{V|U}(\mathcal{G}_1\cap \mathcal{G}_2)\notag\\
&=H(V_i,U_i,i\in \mathcal{G}_1\cup \mathcal{G}_2)+H(V_i,U_i,i\in \mathcal{G}_1\cap \mathcal{G}_2)\notag\\
&\quad-H(U_i,i\in \mathcal{G}_1\cup \mathcal{G}_2)-H(U_i,i\in \mathcal{G}_1\cap \mathcal{G}_2).
\end{align}
The mutual independence among $U_i$'s gives
\begin{align}
&H(U_i,i\in \mathcal{G}_1)+H(U_i,i\in \mathcal{G}_2)\notag\\
&=H(U_i,i\in \mathcal{G}_1\cup \mathcal{G}_2)+H(U_i,i\in \mathcal{G}_1\cap \mathcal{G}_2).
\end{align}
The submodularity of unconditioned entropy function of random variables is well-known \cite{Fujishige}, which gives
\begin{align}
&H(V_i,U_i,i\in \mathcal{G}_1)+H(V_i,U_i,i\in \mathcal{G}_2)\notag\\
&\geq H(V_i,U_i,i\in \mathcal{G}_1\cup \mathcal{G}_2)+H(V_i,U_i,i\in \mathcal{G}_1\cap \mathcal{G}_2)
\end{align}
and the proof is thus complete.
\end{proof}

\section{Proof of the Lemmas}
\label{appendix:lemmas}

\begin{proof}[Proof of Lemma \ref{lemma:increasingijk}]
\begin{align}
&\phi_{i,j}\phi_{j,k}=\binom{K-i}{j-i}^{-1}\binom{j-1}{j-i} \binom{K-j}{k-j}^{-1}\binom{k-1}{k-j}\notag\\
&\quad=\frac{(j-i)!(K-j)!(j-1)!}{(K-i)! (j-i)!(i-1)!}\frac{(k-j)!(K-k)!(k-1)!}{(K-j)!(k-j)!(j-1)!}\notag\\
&\quad=\frac{(K-k)!(k-1)!}{(K-i)!(i-1)!}\frac{(k-i)!}{(k-i)!}=\phi_{i,k}.
\end{align}
\end{proof}

\begin{proof}[Proof of Lemma \ref{lemma:decreasingijk}]
\begin{align}
&\phi_{k,j}\phi_{j,i}=\binom{k}{k-j}^{-1}\binom{K-j}{k-j} \binom{j}{j-i}^{-1}\binom{K-i}{j-i} \notag\\
&\quad=\frac{(k-j)!j!(K-j)!}{k!(k-j)!(K-k)!} \frac{(j-i)!i!(K-i)!}{j!(j-i)!(K-j)!}\notag\\
&\quad=\frac{i!(K-i)!}{k!(K-k)!}\frac{(k-i)!}{(k-i)!}=\phi_{k,i}.
\end{align}
\end{proof}
\begin{proof}[Proof of Lemma \ref{lemma:circleij}]
By the definition of $\phi_{i,j}$, it is easy to verify that 
\begin{align}
&\phi_{i,j}\phi_{j,i}=\frac{i}{j}<1.
\end{align}
\end{proof}
\begin{proof}[Proof of Lemma \ref{lemma:twostep}]
The case $i=k$ is exactly Lemma \ref{lemma:circleij}, thus we only need to consider the case $i\neq k$; we may also assume $j\neq i$ and $j\neq k$ since these cases are trivial. The order of $i,j,k$ can be arbitrary, but since the proof only relies on Lemma \ref{lemma:increasingijk}, \ref{lemma:decreasingijk} and \ref{lemma:circleij}, we may assume without loss of generality $i<j$. Thus we have the only three cases. (1) $k<i<j$:  by Lemma \ref{lemma:decreasingijk} and \ref{lemma:circleij}, we have $\phi_{i,k}>\phi_{i,k}\phi_{i,j}\phi_{j,i}=\phi_{i,j}\phi_{j,k}$. (2) $i<k<j$: by Lemma \ref{lemma:increasingijk} and \ref{lemma:circleij}, we have $\phi_{i,k}>\phi_{i,k}\phi_{k,j}\phi_{j,k}=\phi_{i,j}\phi_{j,k}$. (3) $i<j<k$: the equality is implied by Lemma \ref{lemma:increasingijk}.

\end{proof}

\begin{proof}[Proof of Lemma \ref{lemma:consequtivedown}]
\begin{align}
&(k+1)\binom{K-1}{k-1}\phi_{k+1,j}\notag\\
&=\frac{(k+1)(K-1)!(K-j)!(k+1-j)!j!}{(k-1)!(K-k)!(k+1-j)!(K-k-1)!(k+1)!}\notag\\
&=\frac{(K-1)!(K-j)!j!}{(k-1)!(K-k)!(K-k-1)!k!}.
%=k{K-1 \choose k}\phi_{k,j}
\end{align}
Similarly, we have
\begin{align}
k\binom{K-1}{k}\phi_{k,j}&=\frac{k(K-1)!(K-j)!(k-j)!j!}{k!(K-k-1)!(k-j)!(K-k)!k!}\notag\\
&=\frac{(K-1)!(K-j)!j!}{(k-1)!(K-k)!(K-k-1)!k!},
\end{align}
proving the lemma.
\end{proof}

\begin{proof}[Proof of Lemma \ref{lemma:alternativephijk}]
We only need to write the following
\begin{align}
\binom{K-1}{i-1}{\binom{K-1}{j-1}}^{-1}
&=\frac{(K-1)!(j-1)!(K-j)!}{(i-1)!(K-i)!(K-1)!}\notag\\
&=\frac{(j-1)!}{(i-1)!(j-i)!}\frac{(K-j)!(j-i)!}{(K-i)!}\notag\\
&=\binom{j-1}{j-i}{\binom{K-i}{j-i}}^{-1}=\phi_{i,j}.
\end{align}
\end{proof}

%\begin{proof}[Proof of Lemma \ref{lemma:optimalgeneral}]
%Due to the symmetry, we only need to prove one of the inequalities. Suppose the first inequality is not true, then we show that a new solution can strictly increase the quantity being maximized, which is a contradiction. In the new solution let $r'_{i,j^*}=r_{i,j}$ and $r'_{i,j}=0$ and the other values  remain the same. To see the contradiction, observe
%\begin{align}
%A_{j^*}\phi_{i,j^*}r'_{i,j^*}\geq A_{j^*}\phi_{i,j}\phi_{j,j^*}r'_{i,j^*}>A_j\phi_{i,j}r'_{i,j^*}=A_j\phi_{i,j}r_{i,j^*}\notag
%\end{align}
%where the first inequality is by Lemma \ref{lemma:twostep} and the second inequality is by the supposition. 
%\end{proof}

\begin{proof}[Proof of Lemma \ref{lemma:extremal}]
Suppose an arbitrary optimal solution of the maximization problem (\ref{eqn:maximizingBstar}) is given, we shall next transform it into an extremal one which is also optimal. 

For condition (\textsf{i}), we may assume $R^*_i>0$ because otherwise the statement is trivial. Observe that for any optimal solution, the second inequality in (\ref{eqn:condition2}) must hold with equality, because otherwise the quantity being maximized can strictly increase.
First suppose for certain $i$, there exist distinct $j_1,j_2$ such that $r_{i,j_1}>0$ and $r_{i,j_2}>0$, then we must have
\begin{align}
A_{j_1}\phi_{i,j_1}r_{i,j_1}=A_{j_2}\phi_{i,j_2}r_{i,j_2},\label{eqn:mustbeequal}
\end{align}
because otherwise, e.g., if $<$ held, then letting $r'_{i,j_2}=r_{i,j_1}+r_{i,j_2}$ and $r'_{i,j_1}=0$ strictly increases the quantity being maximized in (\ref{eqn:maximizingBstar}). However, if (\ref{eqn:mustbeequal}) is true, the new solution given above does not decrease the quantity being maximized, thus a new solution can be found such that there exist no such two distinct $j_1,j_2$. Given this is true, it is clear that for each $i$, letting the unique $j$ for which $r_{i,j}>0$ be $R^*_i$ is an optimal choice. Thus condition (\textsf{i}) is satisfied by some optimal solution, and from here on, we shall only consider such solutions.

For condition (\textsf{ii}), we may assume $R^*_j>0$ since otherwise the statement is trivial. %The optimality of the solution implies
%\begin{align}
%A_iR^*_i\leq A_j\phi_{i,j}R^*_i.
%\end{align}
Suppose condition (\textsf{ii}) is not true, i.e., for some $k\neq j$, $r_{j,k}=R^*_j$, then the optimality of the solution implies 
\begin{align}
\label{eqn:jtok}
A_j\leq A_k\phi_{j,k},
\end{align}
Now we claim that the new solution with $r'_{i,k}=R^*_i$ and $r'_{i,j}=0$ (with other $r_{i,j}$ values unchanged) can not decrease the quantity being maximized. To see this, we only need to observe that
\begin{align}
A_k\phi_{i,k}R^*_i\geq A_k\phi_{i,j}\phi_{j,k}R^*_i \geq A_j\phi_{i,j}R^*_i,
\end{align}
which is by Lemma \ref{lemma:twostep} and (\ref{eqn:jtok}). Thus the conditions (\textsf{i}) and (\textsf{ii}) are indeed satisfied by some optimal solution, and from here on we shall only consider such solutions.

For condition (\textsf{iii}), we only discuss the case $i<j$, because the other case $i>j$ is similar. The fact that $r_{i,j}>0$ implies $A_iR^*_i\leq A_j\phi_{i,j}R^*_i$. We may assume $R^*_k>0$ because otherwise the statement is trivial. 
%and from Lemma \ref{lemma:optimalgeneral} we have
%\begin{align}
%A_j\geq A_{j'}\phi_{j,j'},\quad\mbox{and} \quad A_{j'}\geq A_j\phi_{j',j}.
%\end{align}
Take an arbitrary $k$, such that $i<k<j$, we may have $r_{k,j'}=R^*_k$ for some $j'$, and the value of $j'$ may be $j'<i$, $i<j'<k$ or $ j'\leq k$; note that we can assume $j'\neq i$ since condition (\textsf{i}) afore-proved.
It is easy to see that we must have $A_j>0$ and $A_{j'}>0$. The fact that $r_{i,j}>0$ and $r_{k,j'}>0$ imply that
\begin{align}
\label{eqn:twoorders}
A_j\phi_{i,j}\geq A_{j'}\phi_{i,j'},\quad\mbox{and} \quad A_{j'}\phi_{k,j'}\geq A_{j}\phi_{k,j}.
\end{align}
The three cases are now discussed individually next. Case (1) $j'<i$, from (\ref{eqn:twoorders}) and Lemma \ref{lemma:increasingijk} and \ref{lemma:decreasingijk}, we have
\begin{align}
A_j\phi_{i,k}\phi_{k,j}\geq A_{j'}\phi_{i,j'},\quad\mbox{and} \quad A_{j'}\phi_{k,i}\phi_{i,j'}\geq A_{j}\phi_{k,j},
\end{align}
which lead to $\phi_{i,k}\phi_{k,i}\geq 1$, contradicting Lemma \ref{lemma:circleij}, thus this is an impossible case. Case (2) $i<j'<k$, from (\ref{eqn:twoorders}) and Lemma \ref{lemma:increasingijk}, we have 
\begin{align}
A_j\phi_{i,j'}\phi_{j',k}\phi_{k,j}\geq A_{j'}\phi_{i,j'},\quad\mbox{and} \quad A_{j'}\phi_{k,j'}\geq A_{j}\phi_{k,j},
\end{align}
which lead to $\phi_{j',k}\phi_{k,j'}\geq 1$, thus this is another impossible case. 
Case (3) $j'\geq k$, from (\ref{eqn:twoorders}) and Lemma \ref{lemma:increasingijk}, we have that for this case
\begin{align}
A_j\phi_{i,k}\phi_{k,j}\geq A_{j'}\phi_{i,k}\phi_{k,j'}
\end{align}
thus the new solution that $r'_{k,j}=r_{k,j'}$ and $r'_{k,j'}=0$ does not decrease the quantity being optimized. Thus the conditions (\textsf{i}), (\textsf{ii}) and (\textsf{iii}) are indeed satisfied simultaneously by some optimal solution. The lemma is proved.
\end{proof}

\bibliographystyle{IEEEbib}

\end{document}